\documentclass[11pt,a4paper]{amsart}
\usepackage{amsmath,amsfonts,amsthm,amssymb}
\usepackage{color}
\usepackage{epsfig}
\usepackage{graphicx}
\usepackage{cite,url}
\usepackage{hyperref}
\usepackage{xspace}
\usepackage{enumerate}
\usepackage[left=1in,right=1in,top=1in,bottom=1in,foot=0.5in]{geometry}

\theoremstyle{plain}
\newtheorem{theorem}{Theorem}[section]
\newtheorem{lemma}[theorem]{Lemma}
\newtheorem{proposition}[theorem]{Proposition}

\newtheorem{claim}{Claim}
\newtheorem*{claimet}{Claim}
\theoremstyle{definition}

\newtheorem{question}[theorem]{Question}

\newtheorem{remark}[theorem]{Remark}
\newtheorem{conjecture}[theorem]{Conjecture}

\numberwithin{equation}{section}

\DeclareMathOperator{\conv}{conv}

\DeclareMathOperator{\St}{St}
\DeclareMathOperator{\Link}{Link}
\DeclareMathOperator{\dimension}{dim}
\DeclareMathOperator{\degree}{deg}

\newcommand{\bX}{\ensuremath{\mathbf{X}\xspace}}
\newcommand{\bG}{\ensuremath{\mathbf{G}\xspace}}
\newcommand{\bY}{\ensuremath{\mathbf{Y}\xspace}}
\newcommand{\bPi}{\ensuremath{\mathbf{\Pi}\xspace}}

\newcommand{\te}{\ensuremath{\widetilde{e}}\xspace}
\newcommand{\tm}{\ensuremath{\widetilde{m}}\xspace}
\newcommand{\tbetaX}{\ensuremath{\widetilde{\beta X}}\xspace}
\newcommand{\tbetabX}{\ensuremath{\widetilde{\beta \bX}}\xspace}
\newcommand{\tlambda}{\ensuremath{\widetilde{\lambda}}\xspace}
\newcommand{\tnu}{\ensuremath{\widetilde{\nu}}\xspace}
\newcommand{\tildo}{\ensuremath{\widetilde{o}}\xspace}
\newcommand{\tu}{\ensuremath{\widetilde{u}}\xspace}
\newcommand{\tv}{\ensuremath{\widetilde{v}}\xspace}
\newcommand{\tw}{\ensuremath{\widetilde{w}}\xspace}

\newcommand{\ty}{\ensuremath{\widetilde{y}}\xspace}
\newcommand{\tz}{\ensuremath{\widetilde{z}}\xspace}

\newcommand{\tH}{\ensuremath{\widetilde{H}}\xspace}
\newcommand{\tQ}{\ensuremath{\widetilde{Q}}\xspace}
\newcommand{\tV}{\ensuremath{\widetilde{V}}\xspace}
\newcommand{\tW}{\ensuremath{\widetilde{W}}\xspace}
\newcommand{\tX}{\ensuremath{\widetilde{X}}\xspace}
\newcommand{\tY}{\ensuremath{\widetilde{Y}}\xspace}

\newcommand{\tbX}{\ensuremath{\widetilde{\mathbf{X}}\xspace}}
\newcommand{\tbY}{\ensuremath{\widetilde{\mathbf{Y}}\xspace}}

\newcommand{\cD}{\ensuremath{\mathcal{D}\xspace}}
\newcommand{\cE}{\ensuremath{\mathcal{E}\xspace}}
\newcommand{\cF}{\ensuremath{\mathcal{F}\xspace}}
\newcommand{\cH}{\ensuremath{\mathcal{H}\xspace}}
\newcommand{\cM}{\ensuremath{\mathcal{M}\xspace}}

\newcommand{\ZZ}{\mathbb {\mathbb{Z}\xspace}}
\newcommand{\NN}{\mathbb {\mathbb{N}\xspace}}

\begin{document}


%






\title[A Counterexample to Thiagarajan's Conjecture on Regular
  Event Structures]{A Counterexample to Thiagarajan's Conjecture \\ on Regular
  Event Structures$^1$} \thanks{$^1$An extended abstract~\cite{CC-ICALP17} of
  this paper appeared in the proceedings of ICALP 2017.}

\author[J.\ Chalopin]{J\' er\'emie Chalopin}
\address{Laboratoire d'Informatique et Syst\`emes, Aix-Marseille Universit\'e and CNRS
}
\email{jeremie.chalopin@lif.univ-mrs.fr}

\author[V.\ Chepoi]{Victor Chepoi}
\address{Laboratoire d'Informatique et Syst\`emes, Aix-Marseille Universit\'e and CNRS}
\email{victor.chepoi@lif.univ-mrs.fr}

\begin{abstract}  
 We provide a counterexample to a conjecture by Thiagarajan (1996 and
 2002) that regular event structures correspond exactly to event
 structures obtained as unfoldings of finite 1-safe Petri nets.  The
 same counterexample is used to disprove a closely related conjecture
 by Badouel, Darondeau, and Raoult (1999) that domains of regular
 event structures with bounded $\natural$-cliques are recognizable by
 finite trace automata.  Event structures, trace automata, and Petri
 nets are fundamental models in concurrency theory. There exist nice
 interpretations of these structures as combinatorial and geometric
 objects and both conjectures can be reformulated in this
 framework. Namely, from a graph theoretical point of view, the
 domains of prime event structures correspond exactly to median
 graphs; from a geometric point of view, these domains are in
 bijection with CAT(0) cube complexes.

 A necessary condition for both conjectures to be true is that domains
 of regular event structures (with bounded $\natural$-cliques) admit a
 regular nice labeling (which corresponds to a special coloring of the
 hyperplanes of the associated CAT(0) cube complex).  To disprove
 these conjectures, we describe a regular event domain (with bounded
 $\natural$-cliques) that does not admit a regular nice labeling. Our
 counterexample is derived from an example by Wise (1996 and 2007) of
 a nonpositively curved square complex ${\bf X}$ with six squares,
 whose edges are colored in five colors, and whose universal cover
 $\tbX$ is a CAT(0) square complex containing a
 particular plane with an aperiodic tiling.  We prove that other
 counterexamples to Thiagarajan's conjecture arise from aperiodic
 4-way deterministic tile sets of Kari and Papasoglu (1999) and
 Lukkarila (2009).


  On the positive side, using breakthrough results by Agol (2013) and
  Haglund and Wise (2008, 2012) from geometric group theory, we prove
  that Thiagarajan's conjecture is true for regular event structures
  whose domains occur as principal filters of hyperbolic CAT(0) cube
  complexes which are universal covers of finite nonpositively curved
  cube complexes.
\end{abstract}


\maketitle

\section{Introduction}

Event structures, introduced by Nielsen, Plotkin, and Winskel
\cite{NiPlWi,Wi,WiNi}, are a widely recognized abstract model of
concurrent computation. An event structure (or more precisely, a prime
event structure or an event structure with binary conflict) is a
partially ordered set of the occurrences of actions, called events,
together with a conflict relation.  The partial order captures the
causal dependency of events. The conflict relation models
incompatibility of events so that two events that are in conflict
cannot simultaneously occur in any state of the
computation. Consequently, two events that are neither ordered nor in
conflict may occur concurrently. More formally, an event structure is
a triple ${\mathcal E}=(E,\le, \#),$ consisting of a set $E$ of
events, and two binary relations $\le$ and $\#$, the causal dependency
$\le$ and the conflict relation $\#$ with the requirement that the
conflict is inherited by the partial order $\le$. The pairs of events
not in $\le\cup \ge \cup ~\#$ define the concurrency relation $\|$.
The domain of an event structure consists of all computation states,
called configurations.  Each computation state is a subset of events
subject to the constraints that no two conflicting events can occur
together in the same computation and if an event occurred in a
computation then all events on which it causally depends have occurred
too. Therefore, the domain of an event structure ${\mathcal E}$ is the
set $\mathcal D({\mathcal E})$ of all finite configurations ordered by
inclusion. An event $e$ is said to be enabled by a configuration $c$
if $e\notin c$ and $c\cup \{ e\}$ is a configuration. The degree of an
event structure $\mathcal E$ is the maximum number of events enabled
by a configuration of $\mathcal E$. The future (or the principal
filter, or the residual) of a configuration $c$ is the set of all
finite configurations $c'$ containing $c$.

Among other things, the importance of event structures stems from the fact that several
fundamental  models of concurrent computation lead to event structures.
Nielsen, Plotkin, and Winskel \cite{NiPlWi} proved that every 1-safe Petri net $N$ unfolds
into  an event structure ${\mathcal E}_N$. Later results of \cite{NiRoThi}
and \cite{WiNi} show in fact that 1-safe Petri nets and event structures represent each
other in a strong sense. In the same vein, Stark \cite{St} established
that the domains of configurations of trace automata are exactly the conflict event domains;
a presentation of domains of  event structures as trace monoids
(Mazurkiewicz traces)  or as asynchronous transition systems was given in \cite{RoTh} and \cite{Be},
respectively.  In both cases, the events of the resulting
event structure are labeled  (in the case of trace monoids and trace automata---by the letters
of a possibly infinite trace alphabet $M=(\Sigma,I)$)
in a such a way that  any two events enabled by the same configuration are labeled differently
(such a labeling is usually called  a nice labeling).

To deal with {\it finite} 1-safe Petri nets, Thiagarajan
\cite{Thi_regular,Thi_conjecture} introduced the notions of regular
event structure and regular trace event structure.  A regular event
structure $\mathcal E$ is an event structure with a finite number of
isomorphism types of futures of configurations and finite degree.  A
regular trace event structure is an event structure $\mathcal E$ whose
events can be nicely labeled by the letters of a finite trace alphabet
$M=(\Sigma,I)$ in a such a way that the labels of any two concurrent
events define a pair of $I$ and there exists only a finite number of
isomorphism types of labeled futures of configurations.
These definitions were motivated by the fact  that the event
structures ${\mathcal E}_N$ arising from {\it finite} 1-safe Petri
nets $N$ are regular: Thiagarajan \cite{Thi_regular}  proved that
event structures of {\it finite} 1-safe Petri nets correspond to regular trace event structures.
This lead Thiagarajan to formulate the following conjecture:

\begin{conjecture}[\!\!\cite{Thi_regular,Thi_conjecture}]\label{conj-thi}
{A prime event structure $\mathcal E$ is isomorphic to the event structure ${\mathcal E}_N$
arising from a finite 1-safe Petri net $N$  if and only if $\mathcal E$ is regular.}
\end{conjecture}

Badouel, Darondeau, and Raoult \cite{BaDaRa} formulated two similar
conjectures about conflict event domain that are recognizable by
finite trace automata. The first one is equivalent to Conjecture~\ref{conj-thi},
while the second one is formulated in a more general setting with an
extra condition. We formulate their second conjecture in the
particular case of event structures:


\begin{conjecture}[\!\!\cite{BaDaRa}]\label{conj-badara}
{A conflict event domain is recognizable if and only if the  event structure
$\mathcal E$  is regular and has bounded $\natural$-cliques.}
\end{conjecture}

In view of previous results, to establish Conjecture~\ref{conj-thi}, it is
necessary for a regular event structure $\mathcal E$ to have a
regular nice labeling with letters from some trace alphabet
$(\Sigma,I)$.  Nielsen and Thiagarajan \cite{NiThi} proved in a
technically involved but very nice combinatorial way that all regular
conflict-free event structures satisfy Conjecture~\ref{conj-thi}.  In a equally
difficult and technical proof, Badouel et al. \cite{BaDaRa} proved
that their conjectures hold for context-free event domains, i.e., for
domains whose underlying graph is a context-free graph sensu M\"uller
and Schupp \cite{MuSch}.  In this paper, we present a counterexample
to Thiagarajan's Conjecture based on a more geometric and
combinatorial view on event structures. We show that our example also
provides a counterexample to Conjecture~\ref{conj-badara} of Badouel et al.

We use the striking bijections between the domains of event
structures, median graphs, and CAT(0) cube complexes.  Median graphs
have many nice properties and admit numerous characterizations. They
have been investigated in several contexts for more than half a
century, and play a central role in metric graph theory; for more
detailed information, the interested reader can consult the surveys
\cite{BaCh_survey, BaHe}. On the other hand, CAT(0) cube complexes are
central objects in geometric group theory
\cite{Sa,Sa_survey,Wi_raags}. They have been characterized in a nice
combinatorial way by Gromov \cite{Gr} as simply connected cube
complexes in which the links of 0-cubes are simplicial flag complexes.
It was proven in \cite{Ch_CAT,Ro} that 1-skeleta of CAT(0) cube
complexes are exactly the median graphs. Barth\'elemy and Constantin
\cite{BaCo} proved that the Hasse diagrams of domains of event structures
are median graphs and every pointed median graph is the domain of an
event structure.  The bijection between pointed median graphs and event
domains established in \cite{BaCo} can be viewed as the classical
characterization of prime event domains as prime algebraic coherent
partial orders provided by Nielsen, Plotkin, and Winskel
\cite{NiPlWi}.
More recently, this result was rediscovered in
\cite{ArOwSu} in the language of CAT(0) cube complexes.
Via these bijections, the events of an event structure $\cE$ correspond to the
parallelism classes of edges of the domain $D(\cE)$ viewed as a median
graph. We recall these bijections in Section \ref{sec-median-CAT(0)}.

Since in our paper we deal only with regular event structures, we need to be able to
construct regular event domains from CAT(0) cube complexes. By Gromov's characterization,
CAT(0) cube complexes are exactly the universal covers of cube complexes satisfying the
link condition, i.e., of nonpositively curved cube (NPC) complexes. Of particular
importance for us are the CAT(0) cube complexes arising as universal covers of {\it finite}
NPC complexes. In Section \ref{sec-directed-npc}, we  adapt the universal cover construction
to directed NPC  complexes $(Y,o)$ and show that every principal filter of the
directed universal cover $(\tY,\tildo)$ is the domain of an event structure. Furthermore, we show that
if the NPC complex $Y$ is finite, then this event structure is regular. Motivated by this result,
we call an event structure {\it strongly regular} if its domain is the principal filter of the
directed universal cover $(\tY,\tildo)$ of a finite directed NPC complex $(Y,o)$.

Our counterexample to Conjectures \ref{conj-thi} and \ref{conj-badara}
is a strongly regular event structure derived from Wise's
\cite{Wi_thesis,Wi_csc} nonpositively curved square complex $\bX$
obtained from a tile set with six tiles. This counterexample is described in Section
\ref{sec-preuve-cex}. In Section \ref{aperiodic} we also prove that other
counterexamples to Thiagarajan's conjecture arise in a similar way
from any aperiodic 4-deterministic tile set, such as the ones
constructed by Kari and Papasoglu \cite{KaPa} and Lukkarila \cite{Lu}.

On the positive side, in Section \ref{thiagu-special} we prove that
Thiagarajan's conjecture is true for event structures whose domains
arise as principal filters of universal covers of finite special cube
complexes.  Haglund and Wise \cite{HaWi1,HaWi2} detected pathologies
which may occur in NPC complexes: self-intersecting hyperplanes,
one-sided hyperplanes, directly self-osculating hyperplanes, and pairs
of hyperplanes, which both intersect and osculate. They called the NPC
complexes without such pathologies \emph{special}.  The main
motivation for introducing and studying special cube complexes was the
profound idea of Wise that the famous virtual Haken conjecture for
hyperbolic 3-manifolds can be reduced to solving problems about
special cube complexes. In a breakthrough result, Agol
\cite{Agol,Agol_ICM} completed this program and solved the virtual
Haken conjecture using the deep theory of special cube complexes
developed by Haglund and Wise~\cite{HaWi1,HaWi2}. The main ingredient
in this proof is Agol's theorem that finite NPC complexes whose
universal covers are hyperbolic are virtually special (i.e., they
admit finite covers which are special cube complexes).  Using this
result of Agol, we can specify our previous result and show that
Thiagarajan's conjecture is true for strongly regular event structures
whose domains occur as principal filters of hyperbolic CAT(0) cube
complexes that are universal covers of finite directed NPC
complexes. Since context-free domains are hyperbolic, this result can
be viewed in some sense as a partial generalization of the result of
Badouel et al. \cite{BaDaRa}.



To conclude this introductory section, we briefly describe the
construction of our counterexample to Thiagarajan's conjecture. It is
based on Wise's \cite{Wi_thesis,Wi_csc} directed nonpositively curved
square complex $\bX$ with one vertex and six squares, whose edges are
colored in five colors, and whose colored universal cover $\tbX$
contains a particular directed plane with an aperiodic tiling.  The
edges of $\bX$ are partitioned into two classes (horizontal and
vertical edges) and opposite edges of squares are oriented in the same
way. As a result, $\tbX$ is a directed CAT(0) square complex whose
edges are colored by the colors of their images in $\bX$ and are
directed in such a way that all edges dual to the same hyperplane are
oriented in the same way. With respect to this orientation, all
vertices of $\tbX$ are equivalent up to automorphism.
We modify the complex $\bX$ by taking its first barycentric
subdivision and by adding to the middles of the edges of $\bX$
directed paths of five different lengths (tips) in order to encode the
colors of the edges of $\bX$ (and $\tbX$) and to obtain a directed
nonpositively curved square complex $W$. The universal cover $\tW$ of
$W$ is a directed (but no longer colored) CAT(0) square complex.
$\tW$ can be viewed as the first barycentric subdivision of the
support $\tX$ of $\tbX$ in which to each vertex arising from a middle
of an edge of $\tbX$ a tip encoding the color of the original edge is
added.  Since $\tW$ is the universal cover of a finite complex $W$,
$\tW$ has a finite number of equivalence classes of vertices up to
automorphism.  From $\tW$ we derive a domain of a regular event
structure $\tW_{\tv}$ by considering the future of an arbitrary vertex
$\tv$ of $\tbX$. Using the fact that $\tbX$ contains a particular
directed plane with an aperiodic tiling, we prove that $\tW_{\tv}$
does not admit a regular nice labeling, thus $\tW_{\tv}$ is the domain
of a regular event structure not having a regular trace labeling.


\section{Event structures}

\subsection{Event structures and domains}
An {\it event structure} is a triple ${\mathcal E}=(E,\le, \#)$, where

\begin{itemize}
\item $E$ is a set of {\it events},
\item $\le\subseteq E\times E$ is a partial order of {\it causal dependency},
\item $\#\subseteq E\times E$ is a binary, irreflexive, symmetric relation of {\it conflict},
\item $\downarrow \!e:=\{ e'\in E: e'\le e\}$ is finite for any $e\in E$,
\item $e\# e'$ and $e'\le e''$ imply $e\# e''$.
\end{itemize}

What we call here an event structure is usually called a {\it coherent
  event structure}, an {\it event structure with a binary conflict},
or a {\it prime event structure}.  Two events $e',e''$ are {\it
  concurrent} (notation $e'\| e''$) if they are order-incomparable and
they are not in conflict.  The conflict $e'\# e''$ between two
elements $e'$ and $e''$ is said to be {\it minimal} (notation
$e'\#_{\mu} e''$) if there is no event $e\ne e',e''$ such that either
$e\le e'$ and $e\# e''$ or $e\le e''$ and $e\# e'$. We say that $e$ is
an \emph{immediate predecessor} of $e'$ (notation $e\lessdot e'$) if
and only if $e\le e', e\ne e'$, and for every $e''$ if $e\le e''\le
e'$, then $e''=e$ or $e''=e'$.

Given two event structures ${\mathcal E}_1=(E_1,\le_1, \#_1)$ and
${\mathcal E}_2=(E_2,\le_2, \#_2)$, a map $f:E_1\rightarrow E_2$ is an
isomorphism if $f$ is a bijection such that $e\le_1 e'$ iff
$f(e)\le_2 f(e')$ and $e\#_1 e'$ iff $f(e)\#_2 f(e')$ for every
$e,e'\in E_1$. If such an isomorphism exists, then ${\mathcal E}_1$
and ${\mathcal E}_2$ are said to be isomorphic; notation ${\mathcal
  E}_1\equiv {\mathcal E}_2$.

A {\it configuration} of an event structure ${\mathcal E}=(E,\le,\#)$ is any finite
subset $c\subset E$ of events which is {\it conflict-free} ($e,e'\in c$
implies that $e,e'$ are not in conflict) and {\it downward-closed} ($e\in c$ and $e'\le e$ implies
that $e'\in c$)  \cite{WiNi}. Notice that $\varnothing$ is always a
configuration and that $\downarrow \!e$ and $\downarrow \!e\setminus \{ e\}$ are configurations for any
$e\in E$.  The {\it domain} of an event structure is the set
$\mathcal D:=\mathcal D({\mathcal E})$ of all configurations of ${\mathcal E}$ ordered by
inclusion; $(c',c)$ is a (directed) edge of the Hasse diagram of
the poset $({\mathcal D}({\mathcal E}),\subseteq)$ if and only if $c=c'\cup \{ e\}$ for
an event $e\in E\setminus c$.   An event $e$ is said to be {\it enabled}
by a configuration $c$ if $e\notin c$ and $c\cup \{ e\}$ is a configuration.
Denote by $en(c)$ the set of all events enabled at the configuration  $c$.
Two events are called {\it co-initial} if they are both enabled at some
configuration $c$. Note that if $e$ and $e'$ are co-initial, then either
$e \#_{\mu} e'$ or $e \| e'$. It is easy to see that two events $e$ and $e'$
are in minimal conflict $e\#_{\mu}e'$ if and only if $e \# e'$ and  $e$
and $e'$ are co-initial.   The {\it degree} $\degree(\mathcal E)$ of an event
structure $\mathcal E$  is the least positive integer $d$ such that $|en(c)|\le d$
for any configuration $c$ of $\mathcal E$. We  say that $\mathcal E$
has {\it finite degree} if $\degree(\mathcal E)$ is finite. The {\it future}
(or the {\it (principal) filter}) ${\mathcal F}(c)$ of a configuration $c$
is the set of all configurations $c'$ containing $c$:
${\mathcal F}(c) = ~\uparrow\! c:=\{ c'\in {\mathcal D}({\mathcal E}): c\subseteq c'\}$, i.e.,
${\mathcal F}(c)$ is the principal filter of $c$ in the
ordered set $({\mathcal D}({\mathcal E}),\subseteq)$.

For an event structure ${\mathcal E}=(E,\le,\#)$, let $\natural$ be
the least irreflexive and symmetric relation on the set of events $E$
such that $e_1\natural e_2$ if (1) $e_1\| e_2$, or (2) $e_1\#_{\mu}
e_2$, or (3) there exists an event $e_3$ that is co-initial with $e_1$
and $e_2$ at two different configurations such that $e_1\| e_3$ and
$e_2\# _{\mu} e_3$ (see Figure~\ref{fig-becarre} for examples). (If
$e_1\natural e_2$ and this comes from condition (3), then we write
$e_1\natural_{(3)} e_2$.) A $\natural$-{\it clique} is any complete
subgraph of the graph whose vertices are the events and whose edges
are the pairs of events $e_1e_2$ such that $e_1\natural e_2$.

\begin{figure}
  \includegraphics[page=6,scale=0.7]{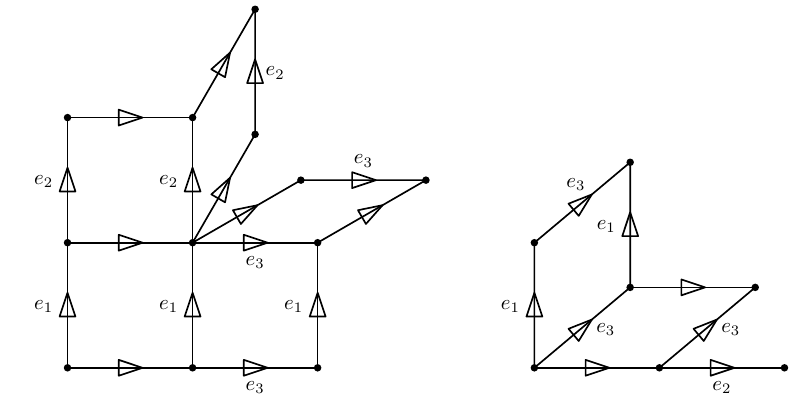}
  \caption{Two examples where $e_1 \natural_{(3)} e_2$: $e_1 \| e_3$ and $e_2
    \#_\mu e_3$}
  \label{fig-becarre}
\end{figure}

A {\it labeled event structure} ${\mathcal E}^{\lambda}=({\mathcal
  E},\lambda)$ is defined by an {\it underlying event structure}
${\mathcal E}=(E,\le, \#)$ and a \emph{labeling} $\lambda$ that is a
map  from $E$ to some alphabet $\Sigma$. Two  labeled
event structures ${\mathcal E}_1^{\lambda_1}=({\mathcal E}_1,\lambda_1)$
and ${\mathcal E}_2^{\lambda_1}=({\mathcal E}_2,\lambda_2)$ are isomorphic
(notation ${\mathcal E}^{\lambda_1}_1\equiv {\mathcal E}^{\lambda_2}_2$) if
there exists an isomorphism $f$ between the underlying event
structures ${\mathcal E}_1$ and ${\mathcal E}_2$ such that
$\lambda_2(f(e_1))=\lambda_1(e_1)$ for every $e_1\in E_1$.

A labeling $\lambda: E\rightarrow \Sigma$ of an event structure
$\mathcal E$ defines naturally a labeling of the directed edges of the
Hasse diagram of its domain $\cD(\cE)$ that we also denote by
$\lambda$. A labeling $\lambda: E\rightarrow \Sigma$ of an event
structure $\mathcal E$ is called a {\it nice labeling} if any two
events that are co-initial have different labels \cite{RoTh}. A nice
labeling of ${\mathcal E}$ can be reformulated as a labeling of the
directed edges of the Hasse diagram of its domain $\mathcal D(\mathcal
E)$) subject to the following local conditions:

\smallskip
\noindent
{\textbf{{Determinism:}}} the  edges outgoing from the same
vertex of ${\mathcal D}({\mathcal E})$ have different labels;

\noindent
{\textbf{{Concurrency:}}} the opposite edges of each square of
${\mathcal D}({\mathcal E})$ are labeled with the same labels.

\smallskip

In the following, we use interchangeably the labeling of an event
structure and the labeling of the edges of its domain.

\subsection{Regular event structures} \label{regular}
In this subsection, we recall the definitions of regular event
structures, regular trace event structures, and regular
nice labelings of event structures. We closely follow the definitions
and notations of \cite{Thi_regular,Thi_conjecture,NiThi}.
Let ${\mathcal E}=(E,\le, \#)$ be an event structure. Let $c$
be a configuration of $\mathcal E$.
Set $\#(c)=\{ e': \exists e\in c, e\# e'\}$. The {\it event structure
rooted at} $c$ is defined to be the triple
$\mathcal E\backslash c=(E',\le',\#')$, where $E'=E\setminus (c\cup \# (c))$, $\le'$
is $\le$ restricted to $E'\times E'$, and $\#'$ is $\#$
restricted to $E'\times E'$. It can be easily seen that the domain
${\mathcal D}({\mathcal E}\backslash c)$ of the event structure
$\mathcal E\backslash c$ is isomorphic to the principal filter
${\mathcal F}(c)$ of $c$ in
${\mathcal D}({\mathcal E})$ such that any configuration $c'$ of
${\mathcal D}({\mathcal E})$ corresponds to the configuration
$c'\setminus c$ of ${\mathcal D}({\mathcal E}\backslash c)$.

For an event structure ${\mathcal E}=(E,\le, \#)$, define
the equivalence relation $R_{\mathcal E}$ on its configurations in the
following way: for two configurations $c$ and $c'$ set $cR_{\mathcal E} c'$
if and only if ${\mathcal E}\backslash c\equiv {\mathcal E}\backslash c'$.
The {\it index} of an event structure $\mathcal E$ is the
number of equivalence classes of $R_{\mathcal E}$, i.e., the
number of isomorphism types of futures of configurations of $\mathcal E$. The event
structure ${\mathcal E}$ is {\it regular} \cite{Thi_regular,Thi_conjecture,NiThi}
if  ${\mathcal E}$ has finite index and finite degree.

Now, let ${\mathcal E}^{\lambda}=(\mathcal E,\lambda)$ be a labeled
event structure. For any configuration $c$ of $\mathcal E$, if we
restrict $\lambda$ to ${\mathcal E}\backslash c$, then we obtain a
labeled event structure $(\mathcal{E}\backslash c,\lambda)$ denoted by
${\mathcal E}^{\lambda}\backslash c$. Analogously, define the
equivalence relation $R_{{\mathcal E}^{\lambda}}$ on its
configurations by setting $c R_{{\mathcal E}^{\lambda}} c'$ if and
only if ${\mathcal E}^{\lambda}\backslash c\equiv {\mathcal
  E}^{\lambda}\backslash c'$. The index of ${\mathcal E}^{\lambda}$ is
the number of equivalence classes of $R_{{\mathcal E}^{\lambda}}$.
We say that an event structure $\mathcal
E$ admits a {\it regular nice labeling} if there exists a nice
labeling $\lambda$ of $\mathcal E$ with a finite alphabet $\Sigma$
such that ${\mathcal E}^{\lambda}$ has finite index.

We continue by recalling the definition of regular trace event
structures from \cite{Thi_regular,Thi_conjecture}. A {\it
  (Mazurkiewicz) trace alphabet} is a pair $M=(\Sigma,I)$, where
$\Sigma$ is a finite non-empty alphabet set and $I\subset \Sigma\times
\Sigma$ is an irreflexive and symmetric relation called the {\it
  independence relation}.  As usual, $\Sigma^*$ is the set of finite
words with letters in $\Sigma$. The independence relation $I$ induces
the equivalence relation $\sim_I$, which is the reflexive and
transitive closure of the binary relation $\leftrightarrow_I$: {\it if
  $\sigma,\sigma'\in \Sigma^*$ and $(a,b)\in I$, then $\sigma
  ab\sigma'\leftrightarrow_I\sigma ba\sigma'$.} The relation $D :=
(\Sigma\times \Sigma)\setminus I$ is called the {\it dependence
  relation}.

An $M$-{\it labeled event structure} is a labeled event structure
${\mathcal E}^{\phi}=(\mathcal E, \lambda)$, where ${\mathcal E}=(E,\le, \#)$
is an event structure and $\lambda: E\rightarrow
\Sigma$ is a labeling function which satisfies the following
conditions:
\begin{enumerate}[{(LES}1)]
\item $e\#_{\mu} e'$ implies $\lambda(e)\ne \lambda (e')$,
\item if $e\lessdot e'$ or $e\#_{\mu} e'$, then $(\lambda(e),\lambda(e'))\in D$,
\item if  $(\lambda(e),\lambda(e'))\in D$, then $e\le e'$ or $e'\le e$ or $e\# e'$.
\end{enumerate}
We  call $\lambda$ a {\it trace labeling} of ${\mathcal E}$ with the
trace alphabet $(\Sigma,I)$. The conditions (LES2) and (LES3)
on the labeling function ensures that the concurrency relation $\|$ of $\mathcal E$ respects
the independence relation $I$ of $M$.  In particular, since $I$ is irreflexive, from (LES3)
it follows that any two concurrent events are labeled differently. Since by (LES1) two events
in minimal conflict are also labeled differently, this implies that $\lambda$ is a finite nice
labeling of $\mathcal E$.

An $M$-labeled event structure ${\mathcal E}^{\lambda}=(\mathcal E, \lambda)$ is {\it regular} if
${{\mathcal E}^{\lambda}}$ has finite index. Finally, an event structure $\mathcal E$ is called a
{\it regular trace event structure} \cite{Thi_regular,Thi_conjecture} iff there exists a trace
alphabet $M=(\Sigma,I)$ and a regular $M$-labeled event structure ${\mathcal E}^{\lambda}$ such
that  $\mathcal E$ is isomorphic to the underlying event structure of ${\mathcal E}^{\lambda}$.
From the definition immediately follows that every regular trace event structure is also a regular
event structure. It turns out that the converse is equivalent to Conjecture~\ref{conj-thi}.
Namely, \cite{Thi_conjecture} establishes the following equivalence (this result dispenses us
from giving a formal definition of 1-safe Petri nets; the interested readers can find it  in
the papers  \cite{Thi_conjecture,NiThi}):

\begin{theorem}[\!\!{\cite[Theorem 1]{Thi_conjecture}}] $\mathcal E$ is a regular trace event structure
if and only if there exists a finite 1-safe Petri net $N$ such that $\mathcal E$ and ${\mathcal E}_N$ are isomorphic.
\end{theorem}

In view of this theorem, Conjecture~\ref{conj-thi} is equivalent to the following conjecture:

\begin{conjecture}\label{conj-thi-2}
{$\mathcal E$ is a regular event structure if and only if $\mathcal E$ is a regular trace event structure.}
\end{conjecture}

Badouel et al.~\cite{BaDaRa} considered recognizable conflict event
domains that are more general than the domains of event structures
we consider in this paper. Since the domain of an event structure
$\mathcal E$ is recognizable if and only if $\mathcal E$ is a regular
trace event structure (see~\cite[Section 5]{Mo}),
Conjecture~\ref{conj-badara} can be reformulated as follows:

\begin{conjecture}\label{conj-badara-2}
  {$\mathcal E$ is a regular event structure iff $\mathcal
    E$ is a regular trace event structure and $\mathcal E$ has bounded
    $\natural$-cliques.}
\end{conjecture}

Since any regular trace labeling is a regular nice labeling, any
regular event structure $\mathcal E$ not admitting a regular nice
labeling is a counterexample to Conjecture~\ref{conj-thi-2} (and thus
to Conjecture~\ref{conj-thi}).  If, additionally, $\mathcal E$ has
bounded $\natural$-cliques, $\mathcal E$ is also a counterexample to
Conjecture~\ref{conj-badara-2} (and thus to
Conjecture~\ref{conj-badara}).

\section{Domains, median graphs, and CAT(0) cube complexes} \label{sec-median-CAT(0)}

In this section, we recall the bijections between domains of event
structures and median graphs/CAT(0) cube complexes established in
\cite{ArOwSu} and \cite{BaCo}, and between median graphs and 1-skeleta
of CAT(0) cube complexes established in \cite{Ch_CAT} and \cite{Ro}.

\subsection{Median graphs}
Let $G=(V,E)$ be a simple, connected, not necessarily
finite graph. The {\it distance}
$d_G(u,v)$ between two vertices $u$ and $v$ is the length of
a shortest $(u,v)$-path, and the {\it interval} $I(u,v)$ between $u$
and $v$ consists of all vertices on shortest $(u,v)$--paths, that
is, of all vertices (metrically) {\it between} $u$ and $v$:
$$I(u,v):=\{ x\in V: d_G(u,x)+d_G(x,v)=d_G(u,v)\}.$$
An induced subgraph of $G$ (or the corresponding vertex set)
is called {\it convex} if it includes the interval of $G$ between
any of its vertices.   A graph
$G=(V,E)$ is {\it isometrically embeddable} into a graph $H=(W,F)$
if there exists a mapping $\varphi : V\rightarrow W$ such that
$d_H(\varphi (u),\varphi (v))=d_G(u,v)$ for all vertices $u,v\in V$.

A graph $G$ is called {\it median} if the interval intersection
$I(x,y)\cap I(y,z)\cap I(z,x)$ is a singleton for each triplet $x,y,z$ of vertices. Median graphs are bipartite.
Basic examples of median graphs are trees, hypercubes, rectangular grids, and
Hasse diagrams of distributive lattices and  of median semilattices~\cite{BaCh_survey}.  With any vertex $v$ of a median
graph $G=(V,E)$ is associated a {\it canonical partial order} $\le_v$ defined by setting $x\le_v y$
if and only if $x\in I(v,y);$ $v$ is called the {\it basepoint} of $\le_v$. Since $G$ is bipartite,
the Hasse diagram $G_v$ of the partial order $(V,\le_v)$ is the graph $G$ in which any edge
$xy$ is directed from $x$ to $y$ if and only if the inequality $d_G(x,v)<d_G(y,v)$ holds. We call
$G_v$ a {\it pointed median graph}.  There is a close relationship between pointed median graphs and
median semilattices. A {\it median semilattice} is a meet semilattice $(P,\le)$ such that (i) for every $x$,
the {\it principal ideal} $\downarrow \!x=\{ p\in P: p\le x\}$ is a distributive lattice, and (ii) any
three elements have a least upper bound in $P$ whenever
each pair of them does.

\begin{theorem}[\!\!\cite{Av}] \label{avann} The Hasse diagram of any median semilattice is a median graph. Conversely,
every median graph defines a median semilattice with respect to any canonical order $\le_v$.
\end{theorem}

Median graphs can be obtained from hypercubes by amalgams and median
graphs are themselves isometric subgraphs of hypercubes \cite{BaVdV,Mu}. The canonical
isometric embedding of a median graph $G$ into a (smallest) hypercube
can be determined by the so called {\it Djokovi\'c-Winkler (``parallelism'')}
relation $\Theta$ on the edges of
$G$~\cite{Dj,Wink}. For median graphs, the
equivalence relation $\Theta$ can be defined as follows. First say
that two edges $uv$ and $xy$ are in relation $\Theta'$ if they are
 opposite edges of a $4$-cycle $uvxy$ in $G$. Then let
$\Theta$ be the reflexive and transitive closure of $\Theta'$. Any equivalence class
of $\Theta$ constitutes a cutset of the median graph $G$, which determines one factor of the
canonical hypercube \cite{Mu}. The cutset (equivalence class) $\Theta(xy)$ containing an edge $xy$
defines a convex split $\{ W(x,y),W(y,x)\}$ of $G$ \cite{Mu}, where $W(x,y)=\{ z\in V:
d_G(z,x)<d_G(z,y)\}$ and $W(y,x)=V\setminus W(x,y)$ (we  call the complementary convex sets $W(x,y)$
and $W(y,x)$ {\it halfspaces}). Conversely, for every convex
split of a median graph $G$ there exists at least one edge $xy$ such
that $\{ W(x,y),W(y,x)\}$ is the given split. We  denote by
$\{ \Theta_i: i\in I\}$ the equivalence classes of the relation
$\Theta$ (in \cite{BaCo}, they were called parallelism classes). For an equivalence
class $\Theta_i, i\in I$, we  denote by
$\{ A_i,B_i\}$ the associated convex split. We  say that  $\Theta_i$ {\it separates} the vertices $x$ and $y$ if $x\in A_i,y\in B_i$ or
$x\in B_i,y\in A_i$. The isometric embedding $\varphi$ of $G$ into a hypercube
is obtained by taking a basepoint $v$, setting $\varphi(v)=\varnothing$ and for any other vertex $u$, letting
$\varphi(u)$ be all parallelism classes of $\Theta$ which separate $u$ from $v$.


We conclude this subsection with the following simple but useful local characterization of convex sets of median graphs
(which holds for much more general classes of graphs):

\begin{lemma} \label{convex} Let $S$ be a connected subgraph of a median graph $G$. Then $S$ is a convex subgraph if and
only if $S$ is locally-convex, i.e., $I(x,y)\subseteq S$ for any two vertices $x,y$ of $S$ having a common neighbor in $S$.
\end{lemma}

\subsection{Nonpositively curved cube complexes}
A 0-cube is a single point. A 1-{\it cube} is an isometric copy of the
segment $[-1,1]$ and has a cell structure consisting of 0-cells
$\{\pm 1\}$ and a single 1-cell. An $n$-{\it cube} is an isometric
copy of $[-1,1]^n$, and has the product structure, so that each
closed cell of $[-1,1]^n$ is obtained by restricting some of the
coordinates to $+1$ and some to $-1$.  A \emph{cube complex} is
obtained from a collection of cubes of various dimensions by
isometrically identifying certain subcubes.  The \emph{dimension} of a
cube complex $X$ is the largest value of $d$ for which $X$ contains a
$d$-cube. A {\it square complex} is a cube complex of dimension 2.
The 0-cubes and the 1-cubes of a cube complex $X$ are called {\it
  vertices} and {\it edges} of $X$ and define the graph $X^{(1)}$, the
\emph{$1$-skeleton} of $X$. We denote the vertices of $X^{(1)}$ by
$V(X)$ and the edges of $X^{(1)}$ by $E(X)$.  For
$i\in\NN$, we denote by $X^{(i)}$ the $i$-{\it
  skeleton} of $X$, i.e., the cube complex consisting of all
$j$-dimensional cubes of $X$, where $j\le i$. A {\it square complex}
$X$ is a combinatorial 2-complex whose 2-cells are attached by closed
combinatorial paths of length 4. Thus, one can consider each 2-cell as
a square attached to the 1-skeleton $X^{(1)}$ of $X$.
The \emph{star} $\St(v,X)$ of a vertex $v$ of $X$ is the
subcomplex spanned by all cubes containing $v$.  The \emph{link} of a vertex $x\in X$
is the simplicial complex $\Link(x)$ with a $(d-1)$-simplex for each $d$-cube containing $x$,
with simplices attached according to the attachments of the corresponding cubes. The link
$\Link(x)$ is said to be a \emph{flag (simplicial) complex}  if each $(d+1)$-clique in $\Link(x)$
spans an $d$-simplex. This flagness  condition of $\Link(x)$ can be restated as follows: whenever
three $(k + 2)$-cubes of ${X}$ share a common $k$-cube containing $x$ and pairwise share common
$(k+1)$-cubes, then they are contained in a $(k+3)$--cube of $X$. A cube complex $X$ is called
{\it simply connected} if it is connected and if every continuous
mapping of the 1-dimensional sphere $S^1$ into $X$ can
be extended to a continuous mapping of the disk $D^2$ with boundary
$S^1$ into $X$.
Note that $X$ is connected iff $G(X)=X^{(1)}$ is connected, and $X$ is
simply connected 
iff $X^{(2)}$ is simply connected. Equivalently, a cube complex $X$ is {simply
connected} if $X$ is connected and every cycle $C$ of
its $1$-skeleton is null-homotopic, i.e., it can be contracted to a
single point by elementary homotopies.

Given two cube complexes $X$ and $Y$, a \emph{covering (map)} is a
surjection $\varphi\colon Y \to X$ mapping cubes to cubes and such that
$\varphi_{|\St(v,Y)}\colon \St(v,Y)\to \St(\varphi(v),X)$ is an
isomorphism for every vertex $v$ in $Y$.  The space $Y$ is then called
a \emph{covering space} of $X$.  A \emph{universal cover} of $X$ is a
simply connected covering space; it always exists and it is unique up
to isomorphism~\cite[Sections 1.3 and 4.1]{Hat}. The universal cover of
a complex $X$ will be denoted by $\tX$. In particular, if $X$ is
simply connected, then its universal cover $\tX$ is $X$ itself.

An important class of cube complexes studied in geometric group theory
and combinatorics is the class of nonpositively curved and CAT(0) cube
complexes. We continue by recalling the definition of CAT(0) spaces.
A {\it geodesic triangle} $\Delta=\Delta (x_1,x_2,x_3)$ in a geodesic
metric space $(X,d)$ consists of three points in $X$ (the vertices
of $\Delta$) and a geodesic  between each pair of vertices (the
sides of $\Delta$). A {\it comparison triangle} for $\Delta
(x_1,x_2,x_3)$ is a triangle $\Delta (x'_1,x'_2,x'_3)$ in the
Euclidean plane  ${\mathbb E}^2$ such that $d_{{\mathbb
E}^2}(x'_i,x'_j)=d(x_i,x_j)$ for $i,j\in \{ 1,2,3\}.$ A geodesic
metric space $(X,d)$ is defined to be a {\it CAT(0) space}
\cite{Gr} if all geodesic triangles $\Delta (x_1,x_2,x_3)$ of $X$
satisfy the comparison axiom of Cartan--Alexandrov--Toponogov:
{\it If $y$ is a point on the side of $\Delta(x_1,x_2,x_3)$ with
vertices $x_1$ and $x_2$ and $y'$ is the unique point on the line
segment $[x'_1,x'_2]$ of the comparison triangle
$\Delta(x'_1,x'_2,x'_3)$ such that $d_{{\mathbb E}^2}(x'_i,y')=
d(x_i,y)$ for $i=1,2,$ then $d(x_3,y)\le d_{{\mathbb
E}^2}(x'_3,y').$}  A geodesic
metric space $(X,d)$ is {\it nonpositively curved} if it is locally CAT(0), i.e.,
any point has a neighborhood  inside which the CAT(0) inequality holds.
CAT(0) spaces can be characterized  in several
different  natural ways and have many strong properties, see for example \cite{BrHa}.
In particular, a geodesic metric space $(X,d)$ is CAT(0) if and only if $(X,d)$ is
simply connected and is nonpositively curved. Gromov \cite{Gr} gave a beautiful
combinatorial characterization of CAT(0) cube complexes, which can be also taken
as their definition:
%

\begin{theorem}[\!\!\cite{Gr}] \label{Gromov} A cube complex $X$ endowed  with the $\ell_2$-metric is CAT(0) if and
only if $X$ is simply connected and the links of all vertices of $X$ are flag complexes.
If $Y$ is a cube complex in which the links of all vertices  are flag complexes, then the universal cover
$\tY$ of $Y$ is a CAT(0) cube complex.
\end{theorem}

In view of the second assertion of Theorem \ref{Gromov}, the cube complexes in which
the links of vertices are flag complexes are called
{\it nonpositively curved cube complexes} or shortly {\it NPC  complexes}.  As a corollary of Gromov's
result, {\it for any NPC complex $X$, its universal cover $\tX$ is CAT(0).}

A square complex $X$ is a $VH${\it -complex}
({\it vertical-horizontal complex}) if the 1-cells (edges) of $X$ are partitioned into
two sets $V$ and $H$ called {\it vertical} and
{\it horizontal} edges respectively, and the edges in each square
alternate between edges in $V$ and $H$. Notice that if $X$ is a $VH$-complex, then $X$ satisfies
the Gromov's nonpositive curvature condition since no three squares
may pairwise intersect on three edges with a common vertex, thus $VH$-complexes are particular NPC square complexes.
A $VH$-complex $X$ is a {\it complete square complex}
(CSC)~\cite{Wi_csc} if any vertical edge and any horizontal edge
incident to a common vertex belong to a common square of $X$. By
\cite[Theorem 3.8]{Wi_csc}, if $X$ is a complete square complex, then
the universal cover $\tX$ of $X$ is isomorphic to the
Cartesian product of two trees. By a {\it plane} $\Pi$ in
$\tX$ we will mean a convex subcomplex of $\tX$
isometric to ${\mathbb R}^2$ tiled by the grid ${\mathbb Z}^2$ into
unit squares.

We continue with the bijection between CAT(0) cube complexes and median graphs:

\begin{theorem}[\!\!\cite{Ch_CAT,Ro}] \label{CAT(0)=median} Median graphs are exactly the 1-skeleta  of CAT(0) cube complexes.
\end{theorem}

The proof of Theorem \ref{CAT(0)=median} presented in \cite{Ch_CAT} is based on the following local-to-global characterization of median graphs:

\begin{theorem}[\!\!\cite{Ch_CAT}] \label{CAT(0)=median1} A graph $G$ is a median graph if and only if its cube complex is simply connected
and $G$ satisfies the 3-cube condition: if three squares of $G$ pairwise intersect in an edge and all three intersect in a vertex,
then they belong to a 3-cube.
\end{theorem}

A \emph{midcube} of the $d$-cube $c$, with $d\geq 1$, is the isometric subspace
obtained by restricting exactly one of the coordinates of $d$ to 0.
Note that a midcube is a $(d-1)$-cube.  The midcubes $a$ and $b$ of $X$ are \emph{adjacent}
if they have a common face, and a \emph{hyperplane} $H$
of $X$ is a subspace that is a maximal connected union of midcubes such that, if
$a,b\subset H$ are midcubes, either $a$ and $b$ are disjoint or
they are adjacent.  Equivalently, a hyperplane $H$ is a maximal connected union of
midcubes such that, for each cube $c$, either $H\cap c=\emptyset$
or $H\cap c$ is a single midcube of $c$.

\begin{theorem}[\!\!\cite{Sa}] \label{Sageev} Each hyperplane $H$ of a CAT(0) cube complex $X$
is a CAT(0) cube complex of dimension at most $\dimension{X}-1$
and  ${X}\setminus H$ consists of exactly two components, called \emph{halfspaces}.
\end{theorem}

A 1-cube $e$ (an edge) is \emph{dual} to the hyperplane $H$ if the
0-cubes of $e$ lie in distinct halfspaces of $X \setminus H$, i.e., if
the midpoint of $e$ is in a midcube contained in $H$.  The relation
``dual to the same hyperplane'' is an equivalence relation on the set
of edges of $X$; denote this relation by $\Theta$ and denote by
$\Theta(H)$ the equivalence class consisting of 1-cubes dual to the
hyperplane $H$ ($\Theta$ is precisely the parallelism relation on the
edges of the median graph $X^{(1)}$).

\subsection{Domains versus median graphs/CAT(0) cube complexes}\label{sec-dom-med}

Theorems 2.2 and 2.3 of Barth\'elemy and Constantin \cite{BaCo} establish
the following bijection between event structures and pointed median graphs
(in \cite{BaCo},  event structures are called sites):

\begin{theorem}[\!\!\cite{BaCo}] \label{median_domain} The (undirected) Hasse diagram
of the domain $({\mathcal D}({\mathcal E}),\subseteq)$ of  any event structure ${\mathcal E}=(E,\le, \#)$ is a
median graph. Conversely, for any median graph $G$ and
any basepoint $v$ of $G$,  the pointed median graph $G_v$ is isomorphic to the Hasse diagram of
the domain of an event structure.
\end{theorem}

The first part of this theorem first establishes that each event domain is a median semilattice (in fact,
the conditions (i) and (ii) of a median semilattice are often taken as the definition of a domain,
see for example,  \cite{BaDaRa,WiIntro})  and follows from Avann's Theorem \ref{avann}.
The bijection between domains of event structures
and median semilattices is equivalent to the bijection between domains of event structures and
prime algebraic coherent partial orders
established in \cite{NiPlWi}. With the help of Theorem \ref{CAT(0)=median1}, we can  provide
an alternative proof of the first part of
Theorem \ref{median_domain}, which we hope can be of independent interest. Since we will use
it further, we also recall the construction
of an event structure from a pointed median graph presented in \cite{BaCo}.

\begin{proof}[Proof of Theorem \ref{median_domain}]
  To prove that the square complex of an event domain ${\mathcal
    D}:={\mathcal D}({\mathcal E})$ is simply connected one has to
  show that any cycle $\sigma$ of the Hasse diagram of $\mathcal D$ is
  0-homotopic.  We proceed by lexicographic induction on the pair
  $(n_1(\sigma),n_2(\sigma))$, where $n_1(\sigma)$ is the maximum
  cardinality of a configuration of $\sigma$ and $n_2(\sigma)$ is the
  number of configurations (vertices) of $\sigma$ of size
  $n_1(\sigma)$.  Let $c$ be a configuration of $\sigma$ of maximum
  size $n_1(\sigma)$. Then the neighbors $c',c''$ of $c$ in $\sigma$
  have cardinality $n_1(\sigma)-1$, say $c'=c\setminus \{ e'\}$ and
  $c''=c\setminus \{ e''\}$. If $e' = e''$, then let $\sigma'$ be the
  obtained from $\sigma$ by removing $c$. If $e' \neq e''$, then the
  set $c_0:=c\setminus \{ e',e''\}$ is conflict-free and downward
  closed, thus $c_0$ is a configuration. As a result, the
  configurations $c,c',c_0,c''$ define a square. In this case, let
  $\sigma'$ be the cycle obtained obtained from $\sigma$ by replacing
  $c$ by $c_0$. Note that there is an elementary homotopy from
  $\sigma$ to $\sigma'$ via the square $cc'c_0c''$. In both cases, if
  $n_2(\sigma)>1$, then $n_1(\sigma')=n_1(\sigma)$ and
  $n_2(\sigma')=n_2(\sigma)-1$. If $n_2(\sigma)=1$, then
  $n_1(\sigma')=n_1(\sigma)-1$. In both cases, by induction hypothesis
  we may assume that $\sigma'$ is 0-homotopic. Since there exists an
  elementary homotopy from $\sigma$ to $\sigma'$ in both cases, the
  cycle $\sigma$ is also 0-homotopic. To show that the graph of
  ${\mathcal D}$ satisfies the 3-cube condition, one can see that
  there exist four possible embeddings of the three squares in
  $\mathcal D$. In each of these cases one can directly conclude that
  the vertex $v$ completing them to a 3-cube must be a configuration
  (see Figure~\ref{fig-3-cube}). Indeed, in the first three cases, the
  set $c(v)$ of events corresponding to this vertex is included in a
  configuration, thus it is conflict--free. It can be also easily seen
  that in all three cases $c(v)$ is downward-closed, i.e., $c(v)$ is a
  configuration. In the last case, $c(v)=\sigma\cup \{ e_1,e_2,e_3\}$.
  Each pair of events of $c(v)$ is contained in one of the
  configurations $\sigma\cup\{e_i,e_j\}, i,j\in \{ 1,2,3\}, i\ne j$,
  whence $c(v)$ is conflict-free.  Pick any $e\in c(v)$.  If $e\in
  \sigma$, then $\downarrow \!e\subset \sigma$. If $e\in \{
  e_1,e_2,e_3\},$ say $e=e_1$, then $\downarrow \!e\subset \sigma\cup
  \{ e_1\}$. In both cases we conclude that $\downarrow \!e\subset
  c(v)$, i.e., $c(v)$ is downward-closed, whence $c(v)$ is a
  configuration.

\begin{figure}
  \includegraphics[page=5,scale=0.75]{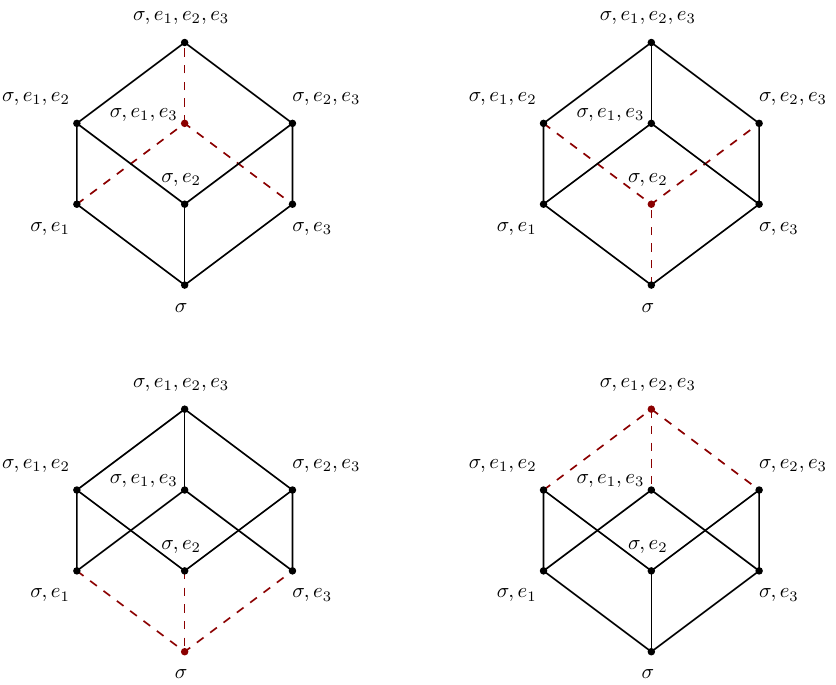}
  \caption{The four possible embeddings of the three squares in  the
    3-cube condition in $\mathcal D$.}
  \label{fig-3-cube}
\end{figure}

 Now, we recall how to define the event structure occurring in the
 second part of the theorem.  Suppose that $v$ is an arbitrary but
 fixed basepoint of a median graph $G.$ For an equivalence class
 $\Theta_i, i\in I,$ we denote by $\{ A_i,B_i\}$ the associated convex
 split, and suppose without loss of generality that $v\in A_i.$ Two
 equivalence classes $\Theta_i$ and $\Theta_j$ are said to be {\it
   crossing} if there exists a 4-cycle $C$ of $G$ with two opposite
 edges in $\Theta_i$ and two other opposite edges in $\Theta_j$
 ($\Theta_i$ and $\Theta_j$ are called \emph{non-crossing} otherwise).
 An equivalence class $\Theta_i$ {\it separates} the basepoint $v$
 from the equivalence class $\Theta_j$ if $\Theta_i$ and $\Theta_j$
 are non-crossing and all edges of $\Theta_j$ belong to $B_i.$ The
 event structure ${\mathcal E}_v=(E,\le,\#)$ associated with a pointed
 median graph $G_v$ is defined in the following way. The set $E$ of
 events is the set $\{ \Theta_i: i\in I\}$ of the equivalence classes
 of $\Theta$. The causal dependency is defined by setting $\Theta_i\le
 \Theta_j$ if and only if $\Theta_i=\Theta_j$ or $\Theta_i$ separates
 $v$ from $\Theta_j$. The conflict relation is defined by setting
 $\Theta_i\#\Theta_j$ if and only if $\Theta_i$ and $\Theta_j$ are
 non-crossing, $\Theta_i$ does not separate $v$ from $\Theta_j$ and
 $\Theta_j$ does not separates $v$ from $\Theta_i.$ Finally, the
 concurrency relation is defined by setting $\Theta_i\| \Theta_j$ if
 and only if $\Theta_i$ and $\Theta_j$ are crossing.  Since each
 parallelism class $\Theta_i$ partitions $G$ into two parts $A_i$ and
 $B_i$, it easily follows that ${\mathcal E}_v=(E,\le,\#)$ satisfies
 the axiom $\Theta_i\# \Theta_j$ and $\Theta_j\le \Theta_k$ imply
 $\Theta_i\# \Theta_k$; consequently ${\mathcal E}_v$ is an event
 structure.  To prove that $G_v$ is the Hasse diagram of the domain
 ${\mathcal D}({\mathcal E}_v)$ of the event structure ${\mathcal
   E}_v$, consider an isometric embedding of $G$ into a hypercube such
 that $v$ corresponds to $\varnothing$. Then any other vertex $u$ of
 $G$ is encoded by a finite set $U$ consisting of all $\Theta_i$ such
 that $\Theta_i$ separates the vertices $v$ and $u$. Since
 $\Theta_i\in U$ and $\Theta_j\le \Theta_i$ implies that $\Theta_j$
 also separates $v$ from $u$, and thus $\Theta_j$ belongs to $U$, we
 conclude that $U$ is downward-closed. If $\Theta_i\# \Theta_j$ and
 $\Theta_i\in U$, then necessarily $\Theta_i$ and $v$ belong to a
 common halfspace defined by $\Theta_j$. Therefore $\Theta_j$ does not
 separate $u$ from $v$. This shows that $U$ is conflict-free, i.e.,
 $U$ is a configuration of ${\mathcal E}_v$. Conversely, any
 configuration $c$ of ${\mathcal E}_v$ consists of exactly those
 $\Theta_i$ that separate $v$ from the vertex representing $c$.  This
 concludes the proof of Theorem \ref{median_domain}.
\end{proof}

Rephrasing the construction of an event structure from a pointed median graph presented
in the proof of Theorem \ref{median_domain}, to each CAT(0) cube complex $X$ and each vertex $v$ of
$X$ one can associate an event structure ${\mathcal E}_v$ such that the domain of ${\mathcal E}_v$
is the 1-skeleton of $X$ pointed at $v$. The events of ${\mathcal E}_v$ are the
hyperplanes of $X$.  Hyperplanes $H$ and $H'$ define concurrent events if and only if they cross,
and $H\leq H'$ if and only if $H=H'$ or $H$ separates $H'$ from $v$.  The events
defined by $H$ and $H'$ are in conflict if and only if $H$ and $H'$ do not cross and neither
separates the other from $v$.  


\subsection{Related work}\label{sec-cex-RT}
The link between event domains, median graphs, and CAT(0) cube
complexes allows a more geometric and combinatorial approach to
several questions on event structures (and to work only with CAT(0)
cube complexes viewed as event domains). For example, this allowed
\cite{Ch_nice} to disprove the so-called {\it nice labeling
  conjecture} of Rozoy and Thiagarajan \cite{RoTh} asserting that any
event structure of finite degree admit a finite nice labeling.  The
topological dimension $\dimension{X}$ of a CAT(0) cube complex $X$
corresponds to the maximum number of pairwise concurrent events of
${\mathcal E}_v$ and to the clique number of the intersection graph of
hyperplanes of $X$. The degree deg$({\mathcal E}_v)$ of the event
structure ${\mathcal E}_v$ is equal to the maximum out-degree of a
vertex in the canonical order $\le_v$ of the 1-skeleton of $X$ (and is
equal to the clique number of a so-called pointed contact graph of
hyperplanes of $X$~\cite{Ch_nice,ChHa}). In particular,
$\dimension{X}\le \degree({\mathcal E}_v)$.  Notice also that the
maximum degree of a vertex of $X$ is upper bounded by
$\degree({\mathcal E}_v)+\dimension(X)\le 2\degree({\mathcal E}_v)$
and is equal to the clique number of the contact graph of hyperplanes
of $X$ (the intersection graph of the carriers of $X$)
\cite{ChHa,Ha}. Using this terminology, a nice labeling of the event
structure ${\mathcal E}_v$ is equivalent to a coloring of the pointed
contact graph of $X$. Using this combinatorial reformulation and the
example of Burling \cite{Bu} of families of axis-parallel boxes of
${\mathbb R}^3$ with no three pairwise intersecting boxes and
arbitrarily high chromatic number of the intersection graph,
\cite{Ch_nice} describes an example of a CAT(0) 4-dimensional cube
complex with maximum degree 12 and infinite chromatic number of the
pointed contact graph, thus providing a counterexample to the nice labeling
conjecture of Rozoy and Thiagarajan \cite{RoTh}. On the other hand, it
is shown in \cite{ChHa} that the nice labeling conjecture is true for
event structures whose domains are 2-dimensional (i.e., event
structures not containing three pairwise concurrent events).

\section{Directed NPC complexes}\label{sec-directed-npc}

Since we can define event structures from their domains, universal
covers of NPC complexes represent a rich source of event structures.
To obtain regular event structures, it is natural to consider
universal covers of finite NPC complexes. Moreover, since domains of
event structures are directed, it is natural to consider universal
covers of NPC  complexes whose edges are directed. However, the
resulting directed universal covers are not in general domains of
event structures. In particular, the domains corresponding to pointed
median graphs given by Theorem~\ref{median_domain} cannot be obtained
in this way. In order to overcome this difficulty, we introduce
directed median graphs and directed NPC complexes. Using these
notions, one can naturally define regular event structures starting
from finite directed NPC complexes.

\subsection{Directed median graphs}

A {\it directed median graph} is a pair $(G,o)$, where $G$ is a median graph and $o$ is
an orientation of the edges of $G$  in a such a way that opposite edges of squares of
$G$ have the same direction. By transitivity of $\Theta$, all edges from the same parallelism
class $\Theta_i$ of $G$ have the same direction. Since each $\Theta_i$ partitions
$G$ into two parts, $o$ defines a partial order $\prec_o$ on the vertex-set of $G$.
For a vertex $v$ of $G$, let ${\mathcal F}_{o}(v,G)=\{ x\in V: v \prec_o x\}$ be the principal
filter of $v$ in the partial order $(V(G),\prec_o)$.  For any canonical basepoint
order $\le_v$ of $G$, $(G,\le_v)$ is  a directed median graph. The converse is obviously
not true: the 4-regular tree $F_4$ directed so that each vertex has two incoming and two
outgoing arcs is a directed median graph which is not induced by a basepoint order.

\begin{lemma} \label{directed-median} For any vertex $v$ of a directed median graph $(G,o)$, the following holds:
\begin{enumerate}[(i)]
\item ${\mathcal F}_{o}(v,G)$ induces a convex subgraph  of $G$;
\item the restriction of the partial order $\prec_o$ on  ${\mathcal F}_{o}(v,G)$ coincides with the restriction of the
canonical basepoint order $\le_v$ on ${\mathcal F}_{o}(v,G)$;
\item ${\mathcal F}_{o}(v,G)$  together with $\prec_o$ is the domain of an event structure;
\item for any vertex $u\in {\mathcal F}_{o}(v,G)$, the principal filter ${\mathcal F}_{o}(u,G)$ is included in ${\mathcal F}_{o}(v,G)$
and ${\mathcal F}_{o}(u,G)$ coincides with the principal filter of $u$ with respect to the canonical basepoint order $\le_v$ on ${\mathcal F}_{o}(v,G)$.
\end{enumerate}
\end{lemma}

\begin{proof}
 To $(i)$: For each parallelism class $\Theta_i$, let $A_i,B_i$ be the
 two convex subgraphs separated by $\Theta_i$ and suppose without loss
 of generality that all edges of $\Theta_i$ are directed from $A_i$ to
 $B_i$.  Let $B$ be the intersection of all  of the
 $B_i$s containing the vertex $v$. We assert that ${\mathcal
   F}_{o}(v,G)$ coincides with $B$.  ${\mathcal F}_{o}(v,G)$ consists
 of all vertices $u$ of $G$ such that there is a path from $v$ to $u$
 in the Hasse diagram of $\prec_o$. Since each $\Theta_i$ is a cutset
 of $G$ and all edges of $\Theta_i$ are directed from $A_i$ to $B_i$,
 we conclude that ${\mathcal F}_{o}(v,G)\subseteq B$. Conversely, let
 $u\in B$ and pick any shortest path
 $P(v,u)=(v_0=v,v_1,\ldots,v_{k-1},v_k=u)$ in $G$ between $v$ and
 $u$. We claim that {\it all edges of $P(v,u)$ are directed from $v$
   to $u$}, yielding $u\in {\mathcal F}_{o}(v,G)$. Pick any edge
 $v_jv_{j+1}$ of $P(v,u)$; suppose that $v_jv_{j+1}$ belongs to the
 parallelism class $\Theta_i$. By convexity of $A_i$ and $B_i$,
 necessarily $\Theta_i$ separates the vertices $v$ and $u$. Since
 $v,u\in B$, this implies that $v\in A_i$ and $u\in B_i$, i.e., the
 edge $v_jv_{j+1}$ is directed from $v_j$ to $v_{j+1}$.

To $(ii)$: First suppose that $u',u\in {\mathcal F}_{o}(v,G)$ and $u'\le_v u$.
This implies that $u'\in I(v,u)$.  Let  $P(v,u)=(v_0=v,v_1,\ldots,v_{k-1},v_k=u)$ be a shortest
path of $G$ between $v$ and $u$ passing via $u'$. By what we have shown in $(i)$,
in the Hasse diagram of $\prec_o$ all the edges of $P(v,u)$ are directed from $v$ to $u$. This
implies that all the edges of the subpath of  $P(v,u)$ comprised between $u'$ and
$u$ are directed from $u'$ to $u$, yielding $u' \prec_o u$. To prove the converse assertion,
suppose by way of contradiction that ${\mathcal F}_{o}(v,G)$ contains two vertices
$u',u$ such that $u' \prec_o u$ however  $u'\le_v u$  is not true, i.e., $u'\notin I(v,u)$.
Among all vertices $u$ for which this holds, suppose that $u$ is chosen so that to
minimize the length of a shortest directed path from $u'$ to $u$.  Let $w$ be a
neighbor of $u$ on a shortest directed path from $u'$ to $u$.  Since $u' \preceq_o w$,
$w\in {\mathcal F}_{o}(v,G)$ and from the choice of $u$ it follows that $u'\le_v w$, i.e., $u'\in I(v,w)$.
Since $G$ is  bipartite, either $w\in I(v,u)$ or $u\in I(v,w)$ holds. If
$w\in I(v,u)$, since $u'\in I(v,w)$, we conclude that $u'\in I(v,u)$,
a contradiction. Therefore
$u\in I(v,w)$. Since the edge $wu$ is oriented from $w$ to $u$ and $u$
lies on a shortest path from $v$ to $w$,  we obtain a contradiction
with the fact that all the edges
of such a shortest path must be directed from $v$ to $w$. This
contradiction establishes that the partial orders $\prec_o$ and $\le_v$ coincide on ${\mathcal F}_{o}(v,G)$.

To $(iii)\&(iv)$: By $(i)$, the subgraph $G'$ of $G$ induced by ${\mathcal F}_{o}(v,G)$
is a median graph. By $(ii)$, the partial order $\prec_o$ coincides on $G'$  with the
canonical basepoint order $\le_v$.
By  Theorem \ref{median_domain}, $(V(G'),\prec_o)$ is the domain of an event structure,
establishing $(iii)$. Finally, $(iv)$ is an immediate consequence of $(ii)$.
\end{proof}

\subsection{Directed NPC cube complexes}
A {\it directed NPC complex} is a pair $(Y,o)$, where $Y$ is a NPC
complex and $o$ is an orientation of the edges of $Y$ in a such a way
that the opposite edges of the same square of $Y$ have the same
direction. Such an orientation $o$ of the edges of a NPC complex $Y$
is called an \emph{admissible} orientation of $Y$. Note that there
exist NPC complexes that do not admit any admissible orientation:
consider a M{\"o}bius band of squares, for example.  An admissible
orientation $o$ of $Y$ induces in a natural way an orientation
$\tildo$ of the edges of its universal cover $\tY$, so that
$(\tY,\tildo)$ is a directed CAT(0) cube complex and
$(\tY^{(1)},\tildo)$ is a directed median graph.  A \emph{directed
  plane} in a directed CAT(0) cube complex $Y$ is a plane $\Pi$ in $Y$
such that for any vertex $(i,j)$ of the grid $\ZZ^2$ tiling $\Pi$,
$(i,j)$ is the source of the edges $(i,j)(i,j+1)$ and $(i,j)(i+1,j)$.

In the following, we need to consider directed colored NPC complexes
and directed colored median graphs. A coloring $\nu$ of a directed NPC
complex $(Y,o)$ is an arbitrary map $\nu: E(Y) \to \Upsilon$ where
$\Upsilon$ is a set of colors. Note that a labeling is a coloring, but
not the converse: labelings are precisely the colorings in which
opposite edges of any square have the same color.  In the following,
we will denote a directed colored NPC complexes by bold letters
like $\bY = (Y,o,\nu)$.  Sometimes, we need to forget the colors and
the orientations of the edges of these complexes. For a complex $\bY$,
we denote by $Y$ the complex obtained by forgetting the colors and the
orientations of the edges of $\bY$ ($Y$ is called the \emph{support}
of $\bY$), and we denote by $(Y,o)$ the directed complex obtained by
forgetting the colors of $\bY$. We also consider directed colored
median graphs that will be the $1$-skeletons of directed colored
CAT(0) cube complexes. Again we will denote such directed colored
median graphs by bold letters like $\bG = (G,o,\nu)$. Note that
(uncolored) directed NPC complexes can be viewed as directed colored
NPC complexes where all edges have the same color.

When dealing with directed colored NPC complexes, we consider only
homomorphisms that preserve the colors and the directions of the
edges. More precisely, $\bY' = (Y',o',\nu')$ is a covering of $\bY =
(Y,o,\nu)$ via a covering map $\varphi$ if $Y'$ is a covering of $Y$
via $\varphi$ and for any edge $e \in E(Y')$ directed from $s$ to $t$,
$\nu(\varphi(e)) = \nu'(e)$ and $\varphi(e)$ is directed from
$\varphi(s)$ to $\varphi(t)$. Since any coloring $\nu$ of a directed
colored NPC complex $Y$ leads to a coloring of its universal cover
$\tY$, one can consider the colored universal cover $\tbY =
(\tY,\tildo,\tnu)$ of $\bY$.

When we consider principal filters in directed colored median graphs
$\bG = (G,o,\nu)$ (in particular, when $G$ is the $1$-skeleton of the
universal cover $\tbY$ of a directed colored NPC complex $\bY$), we
say that two filters are isomorphic if there is an isomorphism between
them that preserves the directions and the colors of the edges.

We now formulate the crucial regularity property of directed colored
median graphs $(\tY^{(1)},\tildo,\tnu)$ when $(Y,o,\nu)$ is
finite.

\begin{lemma} \label{regular-cube}
If $\bY = (Y,o,\nu)$ is a finite directed colored NPC complex, then
$\tbY^{(1)} = (\tY^{(1)},\tildo,\tnu)$ is a directed median
graph with at most $|V(Y)|$ isomorphism types of colored principal
filters.
In particular, if $(Y,o)$ is a finite directed NPC complex, then
$(\tY^{(1)},\tildo)$ is a directed median graph with
at most $|V(Y)|$ isomorphism types of principal filters.
\end{lemma}


\begin{proof}
  Consider a covering map $\varphi : \tbY =
  (\tY,\tildo,\tnu) \to \bY = (Y,o,\nu)$. We first show
  that $(\tY^{(1)},\tildo)$ is a directed median
  graph. By Theorem~\ref{median_domain}, $\tY^{(1)}$ is a
  median graph. Since the image of a square in $\tY^{(1)}$
  is a square in $Y$, since $\varphi$ preserves the direction of the
  edges, and since two opposite edges of a square of $Y$ have the same
  direction, any two opposite edges of a square of
  $\tY^{(1)}$ have the same direction. Consequently,
  $\tY^{(1)}$ is a directed median graph.

  Consider now two vertices $\tu, \tu' \in V(\tY)$ such that
  $\varphi(\tu) = \varphi(\tu')$. In the following, we show that
  $\cF_o(\tu,\tbY^{(1)})$ and $\cF_o(\tu',\tbY^{(1)})$ are isomorphic,
  which implies that there are at most $|V(Y)|$ different isomorphism
  types of colored principal filters by Lemma~\ref{directed-median}.  The
  proof is based on the two following claims. The first claim can be
  easily proved by induction on the length of $P$.

  \begin{claim}\label{claim-rev-chemin}
    For any path $P = (\tu = \tu_0, \tu_1, \ldots, \tu_k)$ in $\tY$,
    there exists a unique path $P' = (\tu' = \tu_0', \tu_1', \ldots,
    \tu_k')$ such that $\varphi(\tu'_i) = \varphi(\tu_i)$ for all $0
    \leq i \leq k$.
  \end{claim}

  \begin{claim}\label{claim-rev-carre}
    For any four paths $P_1 = (\tu = \tu_0, \tu_1, \ldots, \tu_k)$,
    $P_1' = (\tu' = \tu_0', \tu_1', \ldots, \tu_k')$, $P_2 = (\tu =
    \tv_0, \tv_1, \ldots, \tv_\ell)$, and $P_2' = (\tu' = \tv_0',
    \tv_1', \ldots, \tv_\ell')$ in $\tY$ such that $\varphi(\tu'_i) =
    \varphi(\tu_i)$ for all $0 \leq i \leq k$ and $\varphi(\tv'_j) =
    \varphi(\tv_j)$ for all $0 \leq j \leq \ell$, we have $\tu_k =
    \tv_\ell$ if and only if $\tu'_k = \tv'_\ell$.
  \end{claim}

  \begin{proof}
    For each $0 \leq i \leq k$, let $u_i = \varphi(\tu_i) =
    \varphi(\tu_i')$ and for each $0 \leq j \leq \ell$, let $v_j =
    \varphi(\tv_j) = \varphi(\tv_j')$.  Suppose that $\tu_k =
    \tv_\ell$ and consider the cycle $C = P_1 \cdot \overline{P_2} =
    (\tu = \tu_0, \tu_1, \ldots, \tu_k = \tv_\ell, \tv_{\ell-1},
    \ldots, \tv_1, \tv_0 = \tu)$.  Let $n_1(C) = \max\{d(\tu,\tw): \tw
    \in C\}$ and $n_2(C) = |\{\tw: d(\tu,\tw) = n_1(C)\}|$. We prove
    the claim by lexicographic induction on $(n_1(C),n_2(C))$.  If
    $n_1(C) = 0$, then $k = \ell = 0$ and we are done. Suppose now
    that $n_1(C) \geq 1$.

    Suppose that there exists $0 < i < k$ such that $d(\tu, \tu_i) =
    n_1(C)$.  Suppose first that $\tu_{i-1} = \tu_{i+1}$. Then, since
    $\varphi$ is an isomorphism between St$(\tu_i,\tY)$ and
    St$(u_i,Y)$ and between St$(\tu_i',\tY)$ and St$(u_i,Y)$,
    necessarily $\tu_{i-1}' = \tu_{i+1}'$.  By induction hypothesis
    applied to the paths $P_3 = (\tu = \tu_0, \tu_1, \ldots, \tu_{i-1}
    = \tu_{i+1}, \ldots , \tu_k)$, $P'_3 = (\tu' = \tu_0', \tu_1',
    \ldots, \tu_{i-1}' = \tu_{i+1}', \ldots , \tu_k')$, $P_2$, and
    $P_2'$, we have $\tu_k' = \tv_\ell'$ and we are done.  Assume now
    that $\tu_{i-1} \neq \tu_{i+1}$.  Since the graph $\tY^{(1)}$ is
    bipartite, we have $d(\tu, \tu_{i-1}) = d(\tu,\tu_{i+1}) = n_1(C)
    - 1$ and $d(\tu_{i-1}, \tu_{i+1}) = 2$.  Since $\tY^{(1)}$ is
    median, there exists $\tw_i$ such that $d(\tw_i,\tu_{i-1}) =
    d(\tw_i,\tu_{i+1}) = 1$ and $d(\tw_i,\tu)=n_1(C)-2$. Note that
    $\tu_i\tu_{i-1}\tw_i\tu_{i+1}$ is a square in
    $\St(\tu_i,\tY)$. Consequently, since $\varphi$ is an
    isomorphism between St$(\tu_i,\tY)$ and St$(u_i,Y)$ and between
    St$(\tu_i',\tY)$ and St$(u_i,Y)$, there exists $\tw_i'$ such that
    $\tu_i'\tu_{i-1}'\tw_i'\tu_{i+1}'$ is a square in
    $\St(\tu_i',\tY)$ and $\varphi(\tw_i') = \varphi(\tw_i)$.
    By induction hypothesis applied to the paths $P_4 = (\tu = \tu_0,
    \tu_1, \ldots, \tu_{i-1}, \tw_i, \tu_{i+1}, \ldots , \tu_k)$,
    $P'_4 = (\tu' = \tu_0', \tu_1', \ldots, \tu_{i-1}', \tw_i',
    \tu_{i+1}', \ldots , \tu_k')$, $P_2$, and $P_2'$, we have $\tu_k'
    = \tv_\ell'$ and we are done. Analogously, if there exists $0 < j
    < \ell$ such that $d(\tu, \tv_j) = n_1(C)$, we can show that
    $\tu_k' = \tv_\ell'$.

    Suppose now that $n_1(C) = d(\tu,\tu_k = \tv_\ell)$ and $n_2(C) =
    1$. Since $\tY^{(1)}$ is bipartite, $d(\tu, \tu_{k-1}) = d(\tu,
    \tv_{\ell-1}) = n_1(C) - 1$ and $d(\tu_{k-1}, \tv_{\ell-1}) =
    2$. Since $\tY^{(1)}$ is median, there exists $\tw$ such that
    $d(\tw,\tu_{k-1}) = d(\tw,\tv_{\ell-1}) = 1$ and
    $d(\tw,\tu)=n_1(C)-2$. Since $\varphi(\tu_{k-1}) =
    \varphi(\tu'_{k-1})$, there exists a unique neighbor $\tw'$ of
    $\tu'_{k-1}$ such that $\varphi(\tw') = \varphi(\tw)$. Similarly,
    there exists a unique neighbor $\tw''$ of $\tv'_{\ell -1}$ such
    that $\varphi(\tw'') = \varphi(\tw)$. By induction hypothesis
    applied to the paths $P_5 = (\tu = \tu_0, \tu_1, \ldots,
    \tu_{k-1},\tw)$, $P_5' = (\tu' = \tu_0', \tu_1', \ldots,
    \tu_{k-1}',\tw')$, $P_6 = (\tu = \tv_0, \tv_1, \ldots,
    \tv_{\ell-1}, \tw)$, and $P_6' = (\tu' = \tv_0', \tv_1', \ldots,
    \tv_{\ell-1}', \tw'')$, we have that $\tw' = \tw''$. Consequently,
    since $\varphi$ induces a bijection between $\St(\tw,\tY)$
    and $\St(\varphi(\tw),Y)$ and between $\St(\tw',\tY)$
    and $\St(\varphi(\tw'),Y) = \St(\varphi(\tw),Y)$,
    necessarily $\tu'_k = \tv'_\ell$ and we are done.
  \end{proof}

  Define now a map $f_{\tu,\tu'}$ from $V(\tY)$ to $V(\tY)$ such that
  $f_{\tu,\tu'}(\tu) = \tu'$. For any vertex $\tv \in \tY$, consider a
  path $P = (\tu = \tu_0, \tu_1, \ldots, \tu_k=\tv)$ from $\tu$ to
  $\tv$ in $\tY$. By Claim~\ref{claim-rev-chemin}, there exists a
  unique path $P' = (\tu' = \tu_0', \tu_1', \ldots, \tu_k'=\tv')$ such
  that $\varphi(\tu'_i) = \varphi(\tu_i)$ for all $0 \leq i \leq k$
  and we let $f_{\tu,\tu'}(\tv) = \tv'$. Note that $\varphi(f(\tv)) =
  \varphi(\tv)$ and by Claim~\ref{claim-rev-carre}, $f(\tv)$ is
  independent of the choice of the path $P$. Similarly, we can define
  a map $f_{\tu',\tu}$ from $\tY$ to $\tY$ such that
  $f_{\tu',\tu}(\tu')=\tu$ and one can easily see that $f_{\tu,\tu'}
  \circ f_{\tu',\tu} = f_{\tu',\tu} \circ f_{\tu,\tu'} =
  \mathrm{id}$. Consequently, $f_{\tu,\tu'}$ is a bijection from
  $V(\tY)$ to $V(\tY)$.  Moreover, from the definition of
  $f_{\tu,\tu'}$ and from Claim~\ref{claim-rev-chemin}, if
  $\tv_1\tv_2$ is an edge of $\tY$, then
  $f_{\tu,\tu'}(\tv_1)f_{\tu,\tu'}(\tv_2)$ is also an edge of
  $\tY$. Since the orientation of $\tv_1\tv_2$ in $(\tY,\tildo)$ is
  the same as the orientation of the edge
  $v_1v_2=\varphi(\tv_1)\varphi(\tv_2)$ in $(Y,o)$, it is also the
  same as the orientation of $f_{\tu,\tu'}(\tv_1)f_{\tu,\tu'}(\tv_2)$
  in $(\tY,\tildo)$. Furthermore since $\tnu(\tv_1\tv_2) =
  \nu(v_1v_2) = \tnu(f_{\tu,\tu'}(\tv_1)f_{\tu,\tu'}(\tv_2))$,
  $f_{\tu,\tu'}$ preserves the colors of the edges of $\tY$.
  Consequently, $f_{\tu,\tu'}$ is an automorphism of $\tbY =
  (\tY,\tildo,\tnu)$ such that $f_{\tu,\tu'}(\tu) = \tu'$, and
  thus $\cF_o(\tu,\tbY^{(1)})$ and $\cF_o(\tu',\tbY^{(1)})$
  are isomorphic.
\end{proof}

\begin{proposition} \label{regular-finite-npc}
  Consider a finite (uncolored) directed NPC complex $(Y,o)$. Then for
  any vertex $\tv$ of the universal cover $\tY$ of $Y$, the
  principal filter ${\mathcal F}_{\tildo}(\tv,\tY^{(1)})$ with the
  partial order $\prec_{\tildo}$ is the domain of a regular event
  structure with at most $|V(Y)|$ different isomorphism types of
  principal filters.

\end{proposition}

\begin{proof} By Theorem \ref{Gromov}, $\tY$ is a CAT(0) cube complex.
Combining Lemma \ref{directed-median} $(iii)$-$(iv)$ and Lemma  \ref{regular-cube},
we deduce that $({\mathcal F}_{\tildo}(\tv,\tY^{(1)}),\prec_{\tildo})$
is the domain of a regular event structure with at most $|V(Y)|$ different isomorphism types of principal filters.
\end{proof}

We will call an event structure ${\mathcal E}=(E,\le, \#)$ and its domain ${\mathcal D}({\mathcal E})$
{\it strongly regular} if ${\mathcal D}({\mathcal E})$ is isomorphic
to a principal filter of the universal cover of some finite directed NPC complex. In view of
Proposition \ref{regular-finite-npc},
any strongly regular event structure is regular. 

\section{Thiagarajan's conjecture and special NPC complexes}\label{thiagu-special}

\subsection{Special NPC complexes} Consider an NPC complex $Y$, let $\tY$ be its universal
cover and let $\varphi: \tY\rightarrow Y$ be a covering map.
Analogously to CAT(0) cube complexes, one can define the parallelism
relation $\Theta'$ on the set of edges $E(Y)$ of $Y$ by setting that
two edges of $Y$ are in relation $\Theta'$ iff they opposite edges of
a common 2-cube of $Y$. Let $\Theta$ be the reflexive and transitive
closure of $\Theta'$ and let $\{ \Theta_i: i\in I\}$ denote the
equivalence classes of $\Theta$.  For an equivalence class $\Theta_i$,
the hyperplane $H_i$ associated to $\Theta_i$ is the NPC complex
consisting of the midcubes of all cubes of $Y$ containing one edge of
$\Theta_i$. The edges of $\Theta_i$ are \emph{dual} to the hyperplane
$H_i$. Let $\cH(Y)$ be the set of hyperplanes of $Y$.


The hyperplanes of an NPC complex $Y$ do not longer satisfy the nice properties of
the hyperplanes of CAT(0) cube complexes: they do not longer partition the complex in
exactly two parts, they may self-intersect, self-osculate, two hyperplanes may at
the same time cross and osculate, etc. Haglund and Wise \cite{HaWi1} detected five
types of pathologies which may occur in an NPC  complex (see Figure~\ref{fig-special}):
\begin{enumerate}[(a)]
\item self-intersecting hyperplane;
\item one-sided hyperplane;
\item directly self-osculating hyperplane;
\item indirectly self-osculating hyperplane;
\item a pair of hyperplanes, which both intersect and osculate.
\end{enumerate}

\begin{figure}
  \begin{center}
  \includegraphics[page=12,scale=0.75]{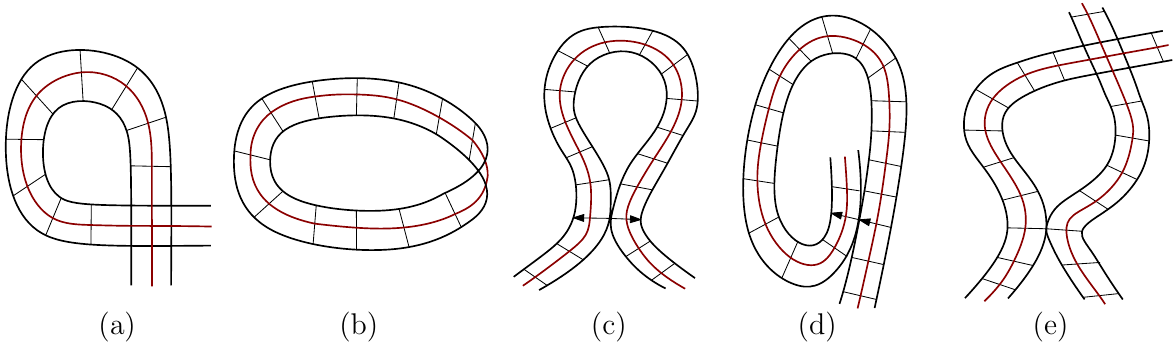}
  \end{center}
  \caption{A self-intersecting hyperplane (a), a one-sided hyperplane (b), a
    directly self-intersecting hyperplane (c), an indirectly
    self-intersecting hyperplane (d), and a pair of hyperplanes that
    inter-osculate (e).}
  \label{fig-special}
\end{figure}

We continue with the definition of each of these pathologies (in which
we closely follow \cite[Section 3]{HaWi1}). Two hyperplanes $H_1$ and
$H_2$ \emph{intersect} if there exists a cube $Q$ and two distinct
midcubes $Q_1$ and $Q_2$ of $Q$ such that $Q_1 \subseteq H_1$ and $Q_2
\subseteq H_2$, i.e., there exists a square with two consecutive edges
$e_1,e_2$ such that $e_1$ is dual to $H_1$ and $e_2$ is dual to $H_2$.

A hyperplane $H$ of $Y$ {\it self-intersects} if it contains more than
one midcube from the same cube, i.e., there exist two edges $e_1,e_2$
dual to $H$ that are consecutive in some square of $Y$ (see Figure~\ref{fig-special}(a)).

A hyperplane $H$ is {\it two-sided} if $N(H)$ is homeomorphic to the
product $H\times (-1,1)$, and there is a combinatorial map
$H\times[-1,1]\rightarrow X$ mapping $H\times \{0\}$ identically to
$H$. The hyperplane is {\it one-sided} if it is not two-sided (see
Figure~\ref{fig-special}(b)). As noticed in \cite[p.1562]{HaWi1},
requiring that the hyperplanes of $Y$ are two-sided is equivalent to
defining an orientation on the dual edges of $H$ such that all sources
of such edges belong to one of the sets $H\times \{-1\}, H\times
\{1\}$ and all sinks belong to the other one. This orientation is
obtained by taking the equivalence relation generated by elementary
parallelism relation: declare two oriented edges $e_1$
and $e_2$ of $Y$ {\it elementary parallel} if there is
a square of $Y$ containing $e_1$ and
$e_2$ as opposite sides and oriented in the same
direction. Notice that if $(Y,o)$ is a directed NPC complex, then
every hyperplane $H$ of $Y$ is two-sided. Conversely, if every
hyperplane $H$ of $Y$ is two-sided, then $Y$ admits admissible
orientations (one can choose an admissible orientation for each
hyperplane independently).

Let $v$ be a vertex of $Y$ and let ${e_1},{e_2}$ be two distinct edges
incident to $v$ but such that ${e_1}$ and ${e_2}$ are not consecutive
edges in some square containing $v$. The hyperplanes $H_1$ and $H_2$
{\it osculate} at $(v,{e_1},{e_2})$ if $e_1$ is dual to $H_1$ and
$e_2$ is dual to $H_2$. The hyperplane $H$ {\it self-osculate} at
$(v,{e_1},{e_2})$ if $e_1$ and $e_2$ are dual to $H$.  Consider a
two-sided hyperplane $H$ and an admissible orientation $o$ of its dual
edges. Suppose that $H$ self-osculate at $(v,{e_1},{e_2})$. If $v$ is
the source of both $e_1$ and $e_2$ or the sink of both $e_1$ and
$e_2$, then we say that $H$ {\it directly self-osculate} at
$(v,{e_1},{e_2})$ (see Figure~\ref{fig-special}(c)). If $v$ is the
source of one of $e_1$, $e_2$, and the sink of the other, then we say
that $H$ {\it indirectly self-osculate} at $(v,{e_1},{e_2})$ (see
Figure~\ref{fig-special}(d)). Note that a self-osculation of a
hyperplane $H$ is either direct or indirect, and this is independent
of the orientation of the edges dual to $H$.

Two hyperplanes $H_1$ and $H_2$ {\it inter-osculate} if they both
intersect and osculate (see Figure~\ref{fig-special}(e)).

Haglund and Wise \cite[Definition 3.2]{HaWi1} called an NPC
complex $Y$ \emph{special} if its hyperplanes are two-sided, do not
self-intersect, do not directly self-osculate, and no two hyperplanes
inter-osculate.


\subsection{Trace labelings of special event structures}
Consider a finite NPC complex $Y$ and let $\cH = \cH(Y)$ be the set of
hyperplanes of $Y$. We define a canonical labeling $\lambda_\cH: E(Y)
\to \cH$ by setting $\lambda_\cH(e) = H$ if the edge $e$ is dual to
$H$. For any covering map $\varphi: \tY \to Y$, $\lambda_\cH$ is
naturally extended to a labeling $\tlambda_{\cH}$ of $E(\tY)$ where
$\tlambda_{\cH}(e) = \lambda_\cH(\varphi(e))$.

We show that strongly regular event structures obtained from finite
special cube complexes admit regular trace labellings.

\begin{proposition}\label{special-trace}
  A finite NPC complex $Y$ with two-sided hyperplanes is special if
  and only if there exists an independence relation $I$ on $\cH =
  \cH(Y)$ such that for any admissible orientation $o$ of $Y$, for any
  covering map $\varphi: \tY \to Y$, and for any principal filter $\cD
  = (\cF_{\tildo}(\tv,\tY^{(1)}),\prec_{\tildo})$ of
  $(\tY,\tildo)$, the canonical labeling $\tlambda_{\cH}$ is a
  regular trace labeling of $\cD$ with the trace alphabet $(\cH,I)$.
\end{proposition}

\begin{proof}
  Suppose first that there exists an independence relation $I
  \subseteq \cH^2$ such that for any admissible orientation $o$ of
  $Y$, for any covering map $\varphi: \tY \to Y$, and any principal
  filter $\cD$ of $(\tY,\tildo)$, $\tlambda_{\cH}$ is a regular
  trace labeling of $\cD$ with the trace alphabet $(\cH,I)$.

  If $Y$ contains a self-intersecting hyperplane $H$, then there exist
  a square $Q$ such that the four edges of $Q$ are dual to
  $H$. Consider an admissible orientation $o$ of $Y$ and note that
  there exist two edges $e_1, e_2$ in $Q$ that have the same source
  $v$. In $(\tY,\tildo)$, consider a vertex $\tv \in
  \varphi^{-1}(v)$ and note that $\tv$ has two outgoing edges $\te_1,
  \te_2$ such that $\varphi(\te_1) = e_1$ and $\varphi(\te_2) =
  e_2$. Since $\tlambda_\cH(\te_1) = \lambda_\cH(e_1) = \lambda_\cH(e_2)
  = \tlambda_\cH(\te_2)$, the labeling $\tlambda_\cH$ violates the
  determinism condition in the principal filter
  $(\cF_{\tildo}(\tv,\tY^{(1)}),\prec_{\tildo})$.

  If $Y$ contains a hyperplane $H$ that directly self-osculate at
  $(v,{e_1},{e_2})$, then there exists an orientation $o$ of $Y$ such
  that $e_1$ and $e_2$ have the same source $v$. In $(\tY,\tildo)$,
  consider a vertex $\tv \in \varphi^{-1}(v)$ and note that $\tv$ has
  two outgoing edges $\te_1, \te_2$ such that $\varphi(\te_1) = e_1$
  and $\varphi(\te_2) = e_2$. Since $\tlambda_\cH(\te_1) =
  \lambda_\cH(e_1) = \lambda_\cH(e_2) = \tlambda_\cH(\te_2)$, the
  labeling $\tlambda_\cH$ violates the determinism condition in the
  principal filter
  $(\cF_{\tildo}(\tv,\tY^{(1)}),\prec_{\tildo})$.



  Finally if $Y$ contains two hyperplanes $H_1$ and $H_2$ that
  inter-osculate, then they osculate at $(v,e_1,e_2)$ and they
  intersect on a square $Q$. We can choose an orientation $o$ of $Y$
  such that $v$ is the source of both $e_1$ and $e_2$. Then there
  exists a source $u$ in $Q$ that has two outgoing edges $e_1'$ and
  $e_2'$ that are parallel respectively to $e_1$ and $e_2$. Let $\tu
  \in \varphi^{-1}(u)$ and $\tv \in \varphi^{-1}(v)$. Let $\te_1,
  \te_2$ be the respective preimages of the edges $e_1,e_2$ such that
  $\tv$ is the source of $\te_1,\te_2$. Similarly, let $\te_1',\te_2'$
  be the respective preimages of the edges $e_1,e_2$ such that $\tu$
  is the source of $\te_1',\te_2'$. Note that $\tlambda_\cH(\te_1) =
  \lambda_\cH(e_1) = \lambda_\cH(e_1') = \tlambda_\cH(\te_1')$ and
  $\tlambda_\cH(\te_2) = \lambda_\cH(e_2) = \lambda_\cH(e_2') =
  \tlambda_\cH(\te_2')$. Consider the principal filters $\cD_v =
  (\cF_{\tildo}(\tv,\tY^{(1)}),\prec_{\tildo})$ and $\cD_u =
  (\cF_{\tildo}(\tu,\tY^{(1)}),\prec_{\tildo})$. In $\cD_v$,
  $\te_1$ and $\te_2$ correspond to two events that are in minimal
  conflict, and thus the pair $(\tlambda_\cH(\te_1),
  \tlambda_\cH(\te_2))$ does not belong to the independence relation
  $I$. On the other hand, in $\cD_u$, $\te_1'$ and $\te_2'$ belong to
  a square, and thus they correspond to two concurrent
  events. Consequently, the pair $(\tlambda_\cH(\te_1'),
  \tlambda_\cH(\te_2'))$ belongs to $I$. Since $\tlambda_\cH(\te_1) =
  \tlambda_\cH(\te_1')$ and $\tlambda_\cH(\te_2) =
  \tlambda_\cH(\te_2')$, we have a contradiction.




 Conversely,  suppose that $Y$ is a finite special NPC complex. We define the
  independence relation $I \subseteq \cH \times \cH$ as follows: $(H_1,H_2) \in
  I$ if and only if the hyperplanes $H_1$ and $H_2$ intersect. From
  its definition, the binary relation $I$ is symmetric. Since no
  hyperplane of $Y$ self-intersects, $I$ is also irreflexive, and thus
  $(\cH,I)$ is a finite trace alphabet.

  Consider an admissible orientation $o$ of $Y$, a vertex $\tv \in
  V(\tY)$, a covering map $\varphi: \tY \to Y$ and consider the
  principal filter $\cD =
  (\cF_{\tildo}(\tv,\tY^{(1)}),\prec_{\tildo})$. By
  Proposition~\ref{regular-finite-npc}, $\cD$ is the domain of a
  regular event structure $\cE$. As explained in
  Subsection~\ref{sec-dom-med}, the events of $\cE$ are the
  hyperplanes of $\cD$.  Hyperplanes $\tH$ and $\tH'$ are concurrent
  if and only if they cross, and $\tH \leq \tH'$ if and only if
  $\tH=\tH'$ or $\tH$ separates $\tH'$ from $v$.  The events $\tH$ and
  $\tH'$ are in conflict iff $\tH$ and $\tH'$ do not cross
  and neither separates the other from $v$. Note that this implies that
  $\tH \lessdot \tH'$ iff $\tH$ separate $\tH'$ from $v$ and $\tH$ and
  $\tH'$ osculate, and $\tH \#_\mu \tH'$ iff $\tH$ and $\tH'$ osculate and
  neither of $\tH$ and $\tH'$ separates the other from $v$. Notice also
  that each hyperplane $\tH'$ of $\cD$ is the intersection of a hyperplane
  $\tH$ of $\tY$ with $\cD$.


  We show that $\tlambda_\cH$ is a regular trace labeling of $\cD$
  with the trace alphabet $(\cH,I)$. First note that if $\te_1,\te_2$
  are opposite edges of a square of $\cD$, then $e_1 = \varphi(\te_1)$
  and $e_2 = \varphi(\te_2)$ are opposite edges of a square of $Y$ and
  thus $\tlambda_\cH(\te_1) = \lambda_\cH(e_1) = \lambda_\cH(e_2) =
  \tlambda_\cH(\te_2)$. Consequently, $\tlambda_\cH$ is a labeling of
  the edges of $\cD$. From Lemma~\ref{regular-cube}, $\cD$ has at most
  $|V(Y)|$ isomorphism types of colored principal filters. Therefore,
  in order to show that $\tlambda_\cH$ is a regular trace labeling of
  $\cD$, we just need to show that $\tlambda_\cH$ satisfies the
  conditions (LES1),(LES2), and (LES3).

  For any two hyperplanes $\tH_1, \tH_2$ in minimal conflict in $\cD$,
  there exist an edge $\te_1$ dual to $\tH_1$ and an edge $\te_2$
  dual to $\tH_2$ such that $\te_1$ and $\te_2$ have the same source
  $\tu$. Note that since $\tH_1$ and $\tH_2$ are in conflict, $\te_1$
  and $\te_2$ do not belong to a common square of $\cD$. Moreover, if
  $\te_1$ and $\te_2$ are in a square $\tQ$ in $\tY$, then since there
  is a directed path from $\tv$ to $\tu$, and since $\tu$ is the
  source of $\tQ$, all vertices of $\tQ$ are in
  $(\cF_{\tildo}(\tv,\tY^{(1)}),\prec_{\tildo}) =
  \cD$. Consequently, the hyperplane $\tH_1$ and $\tH_2$ osculate at
  $(\tu,\te_1,\te_2)$ in $\tY$. Let $u = \varphi(\tu)$, $e_1 =
  \varphi(\te_1)$, and $e_2 = \varphi(\te_2)$, and note that $u$ is
  the source of $e_1$ and $e_2$. Let $H_1$ and $H_2$ be the
  hyperplanes of $Y$ that are respectively dual to $e_1$ and
  $e_2$. Since $\varphi$ is a covering map, $e_1$ and $e_2$ do not
  belong to a common square. Consequently, $H_1$ and $H_2$ osculate at
  $(u,e_1, e_2)$. If $H_1 = H_2$, $H_1$ directly self-osculates at
  $(u,e_1,e_2)$, which is impossible because $Y$ is
  special. Consequently, $\tlambda_\cH(\te_1) = \lambda_\cH(e_1) =
  H_1$ is different from $\tlambda_\cH(\te_2) = \lambda_\cH(e_2) =
  H_2$, establishing (LES1).  Moreover, since no two hyperplanes of
  $Y$ inter-osculate, we know that $H_1$ and $H_2$ do not intersect,
  and thus $(H_1,H_2)\notin I$, establishing (LES2) when $\tH_1
  \#_\mu \tH_2$.

  Suppose now that $\tH_1 \lessdot \tH_2$ in $\cD$. There exist an
  edge $\te_1$ dual to $\tH_1$ and an edge $\te_2$ dual to $\tH_2$ such
  that the sink $\tu$ of $\te_1$ is the source of $\te_2$. Since
  $\tH_1$ separates $\tH_2$ from $\tv$ in $\cD$, $\tH_1$ also
  separates $\tH_2$ from $\tv$ in $\tY$. Consequently, $\te_1$ and
  $\te_2$ do not belong to a common square of $\tY$ and the
  hyperplanes $\tH_1$ and $\tH_2$ osculate at $(\tu,\te_1,\te_2)$. Let
  $u = \varphi(\tu)$, $e_1 = \varphi(\te_1)$, and $e_2 =
  \varphi(\te_2)$, and note that $u$ is the sink of $e_1$ and the
  source of $e_2$. Let $H_1$ and $H_2$ be the hyperplanes of $Y$ that
  are respectively dual to $e_1$ and $e_2$. Since $\varphi$ is a
  covering map, $e_1$ and $e_2$ do not belong to a common
  square. Consequently, $H_1$ and $H_2$ osculate at $(u,e_1,e_2)$. If
  $H_1 = H_2$, then since $I$ is irreflexive, $(H_1,H_1) \notin I$. If
  $H_1 \neq H_2$, since no two hyperplanes of $Y$ inter-osculate, we
  know that $H_1$ and $H_2$ do not intersect, and thus
  $(H_1,H_2)\notin I$, establishing (LES2) when $\tH_1 \lessdot
  \tH_2$.

  We prove (LES3) by contraposition. Consider two hyperplanes $\tH_1,
  \tH_2$ that are concurrent, i.e., they intersect in $\cD$. Since
  $\tH_1$ and $\tH_2$ intersect in $\tY$, there exists a square $\tQ$
  containing two consecutive edges $\te_1, \te_2$ that are
  respectively dual to $\tH_1, \tH_2$.  Let $H_1$ and $H_2$ be the
  hyperplanes of $Y$ that are respectively dual to $e_1
  =\varphi(\te_1)$ and $e_2 = \varphi(\te_2)$. Note that
  $\tlambda_\cH(\te_1) = H_1$ and $\tlambda_\cH(\te_2) = H_2$. Since
  $\varphi$ is a covering map, $e_1$ and $e_2$ belong to a square in
  $Y$.  Then $H_1$ and $H_2$ intersect, and therefore $(H_1,H_2) \in
  I$, establishing (LES3).
\end{proof}

A finite NPC complex $X$ is called {\it virtually special}
\cite{HaWi1,HaWi2} if $X$ admits a finite special cover, i.e., there
exists a finite special NPC complex $Y$ and a covering map $\varphi:
Y\rightarrow X$.  We will call a strongly regular event structure
${\mathcal E}=(E,\le, \#)$ and its domain ${\mathcal D}({\mathcal E})$
{\it cover-special} if ${\mathcal D}({\mathcal E})$ is isomorphic to a
principal filter of the universal cover of some virtually special
complex with an admissible orientation.

\begin{theorem} \label{virtuallyspecial}
  Any cover-special event structure $\cE$ admits a regular trace
  labelling, i.e., Thiagarajan's conjecture is true for cover-special
  event structures.
\end{theorem}

\begin{proof} Let $\cD={\mathcal D}({\mathcal E})$ be the domain of
  $\mathcal E$ and suppose that $\cD$ is the principal filter $\cD =
  (\cF_{\tildo}(\tv,\tX^{(1)}),\prec_{\tildo})$ of
  $(\tX,\tildo)$ for a virtually special complex $X$ and an
  admissible orientation $o$ of its edges.  Let $Y$ be a finite
  special cover of $X$ and let $\varphi: Y\rightarrow X$ be a covering
  map. Let $o'$ be the orientation of the edges of $Y$ obtained from
  $o$ via $\varphi$.  Note that $(X,o)$ and $(Y,o')$ have the same
  universal cover $(\tX,\tildo) = (\tY,\tildo')$


  In particular, the principal filter $\cD =
  (\cF_{\tildo}(\tv,\tX^{(1)}),\prec_{\tildo})$ of
  $(\tX,\tildo)$ is the principal filter
  $(\cF_{\tildo'}(\tv,\tY^{(1)}),\prec_{\tildo'})$ of
  $(\tY,\tildo')$. Since $Y$ is finite and special, by Proposition
  \ref{special-trace} there exists an independence relation $I$ on the
  hyperplanes $\cH=\cH(Y)$ of $Y$ such that the canonical labeling
  $\tlambda_{\cH}$ of $\cD =
  (\cF_{\tildo'}(\tv,\tY^{(1)}),\prec_{\tildo'})$ is a regular
  trace labeling with the trace alphabet $(\cH,I)$.  Therefore,
  Thiagarajan's conjecture holds for the event domain $\cD$.
\end{proof}

\subsection{Strongly hyperbolic regular  event structures}
In this subsection, we show that Thiagarajan's conjecture holds for a
large and natural class of strongly regular event structures, namely those
arising from hyperbolic CAT(0) cube complexes. It turns out that
strongly hyperbolic regular event structures are cover-special. This
is a consequence of the solution by Agol \cite{Agol} of the virtual
Haken conjecture for hyperbolic 3-manifolds. This breakthrough result
of Agol is based on the theory of special cube complexes developed by
Haglund and Wise \cite{HaWi1,HaWi2}.

Similarly to nonpositive curvature, Gromov hyperbolicity is defined in
metric terms. However, as for the CAT(0) property, the
hyperbolicity of a CAT(0) cube complex can be expressed in a purely
combinatorial way.  A metric space $(X,d)$ is $\delta$-{\it
  hyperbolic} \cite{BrHa,Gr} if for any four points $v,w,x,y$ of $X$,
the two largest of the distance sums $d(v,w)+d(x,y)$, $d(v,x)+d(w,y)$,
$d(v,y)+d(w,x)$ differ by at most $2\delta \geq 0$. A graph $G=(X,E)$
endowed with its standard graph-distance $d_G$ is
$\delta$-\emph{hyperbolic} if the metric space $(X,d_G)$ is
$\delta$-{hyperbolic}.  In case of geodesic metric spaces and graphs,
$\delta$-hyperbolicity can be defined in other equivalent ways, e.g.,
via thin or slim geodesic triangles.  For example, a geodesic metric
space $(X,d)$ is $2\delta$-hyperbolic, if all geodesic triangles
$\Delta(x,y,z)$ of $(X,d)$ are $\delta$-{\it slim}, i.e., for any
point $u$ on the side $[x,y]$ the distance from $u$ to $[x,z]\cup
[z,y]$ is at most $\delta$. This definition expresses the negative
curvature of a geodesic metric space.  A metric space $(X,d)$ is {\it
  hyperbolic} if there exists $\delta<\infty$ such that $(X,d)$ is
$\delta$-hyperbolic.  In case of median graphs, i.e., of 1-skeletons
of CAT(0) cube complexes, the hyperbolicity can be characterized in
the following way:

\begin{lemma}[\!\!\cite{ChDrEsHaVa,Ha}] \label{hyp_median}  Let $X$ be a CAT(0) cube complex.
Then its $1$-skeleton $X^{(1)}$ is hyperbolic if and only if all
isometrically embedded square grids are uniformly bounded.
\end{lemma}

We call an event structure ${\mathcal E}=(E,\le, \#)$ and its domain
${\mathcal D}({\mathcal E})$ {\it hyperbolic} if ${\mathcal
  D}({\mathcal E})$ is isomorphic to a principal filter of a directed
CAT(0) cube complex, whose 1-skeleton is hyperbolic.  We call an event
structure ${\mathcal E}=(E,\le, \#)$ and its domain ${\mathcal
  D}({\mathcal E})$ {\it strongly hyperbolic regular} if there exists
a finite directed NPC complex $(X,o)$ such that $\tX$ is hyperbolic
and $\cD$ is a principal filter of $(\tX^{(1)},\tildo)$. Note that an
event structure can be strongly regular and hyperbolic without being
strongly regular hyperbolic (see Remark~\ref{rem-tore-bazar}).

Hyperbolic CAT(0) cube complexes with uniformly bounded degrees have
several strong and nice properties.  It was shown in \cite{Hag} that
such CAT(0) cube complexes can be isometrically embedded into the
Cartesian product of finitely many trees.  Analogously to the nice
labeling conjecture of \cite{RoTh}, a similar result does not hold for
general CAT(0) cube complexes of uniformly bounded degrees
\cite{ChHa}. Modifying the arguments of \cite{Hag} it can be shown
that hyperbolic event structures with bounded degrees admit finite
nice labelings (these labelings are not necessarily regular).  Again
this does not hold for general event structures (see
Subsection~\ref{sec-cex-RT}).

The main result of this section is based on the following very deep
and important result of Agol~\cite{Agol}, following much work of
Haglund and Wise \cite{HaWi1,HaWi2}.  Agol's original result is
formulated in group-theoretical terms. Its following reformulation
(see, for example, \cite[Theorem 6.7]{Br_survey}) in the particular
case of finite NPC complexes is particularly appropriate for our
purposes:

\begin{theorem}[\!\!\cite{Agol}] \label{Agol} Let $X$ be a finite nonpositively
curved cube complex. If the fundamental group $\pi_1(X)$ of $X$ is hyperbolic,
then $X$ is virtually special.
\end{theorem}

The condition that  $\pi_1(X)$ is hyperbolic is equivalent to the fact that the
universal cover $\tX$ of $X$ is hyperbolic. Indeed, it is well-known that
$\pi_1(X)$ acts properly  by deck transformations
on $\tX$; see \cite{Hat} and \cite[Remark 8.3(2)]{BrHa}. Since $X$ is finite, this  action of  $\pi_1(X)$ on $\tX$ is cocompact.
Consequently, $\pi_1(X)$ acts properly and cocompactly by isometries on $\tX$. By
\v{S}varc-Milnor lemma \cite[Proposition 8.19]{BrHa}, the Cayley graph of $\pi_1(X)$ is
quasi-isometric to $\tX$. Since hyperbolicity  is an invariant of quasi-isometry \cite[Theorem 1.9]{BrHa},
$\pi_1(X)$ is hyperbolic if and only if $\tX$ is hyperbolic. Therefore,  any finite NPC complex $X$ that has a
hyperbolic universal cover is virtually special.


\begin{theorem} \label{hyperbolic}
  Any strongly hyperbolic regular event structure admits a regular trace
  labeling, i.e., Thiagarajan's conjecture is true for strongly hyperbolic
  regular event structures.
\end{theorem}

\begin{proof} Let $\cD={\mathcal D}({\mathcal E})$ be the domain of a
  strongly hyperbolic regular event structure $\mathcal E$. Consider a
  finite NPC complex $(X,o)$ such that $\tX$ is hyperbolic and $\cD$
  is the principal filter
  $(\cF_{\tildo}(\tv,\tX^{(1)}),\prec_{\tildo})$ for some $\tv
  \in \tX$.  By Theorem \ref{Agol} of Agol, finite NPC complexes with
  hyperbolic universal covers are virtually special, thus $\mathcal E$
  is a cover-special event structure. By Theorem \ref{virtuallyspecial}, 
  Thiagarajan's conjecture is true for
  $\mathcal E$.
\end{proof}

\section{Wise's event domain $(\tW_{\tv},\prec_{\tildo^*})$}\label{sec-preuve-cex}

In this section, we construct the domain $(\tW_{\tv},\prec_{\tildo^*})$ of a
regular event structure (with bounded $\natural$-cliques) that does
not admit a regular nice labeling. To do so, we start with a directed
colored CSC (complete square complex) $\bX$ introduced by
Wise~\cite{Wi_csc}.  
Recall that in such complexes,  the edges are classified  vertical or horizontal,
 each edge has an orientation and a color, and any two incident edges belong to a square.

\subsection{Wise's square complex $\bX$ and its universal cover $\tbX$}

The complex $\bX$ consists of six squares as indicated in
Figure~\ref{fig-squares} (reproducing Figure~3 of
\cite{Wi_csc}). Each square has two vertical and two horizontal
edges. The horizontal edges are oriented from left to right and
vertical edges from bottom to top.  Denote this orientation of edges
by $o$. The vertical edges of squares are colored white, grey, and
black and denoted $a,b$, and $c$, respectively.  The horizontal edges
of squares are colored by single or double arrow, and denoted $x$ and
$y$, respectively. The six squares are glued together by identifying
edges of the same color and respecting the directions to obtain the
square complex $\bX$. Note that $\bX$ has a unique vertex,
five edges, and six squares. It can be directly
checked that $\bX$ is a complete square complex, and consequently
$(X,o)$ is a directed NPC complex.  Let
$H_{X}$ denote the subcomplex of $X$ consisting of the 2 horizontal
edges and let $V_{X}$ denote the subcomplex of $X$ consisting of the 3
vertical edges.

\begin{figure}
  \begin{center}
  \includegraphics[page=1,scale=0.75]{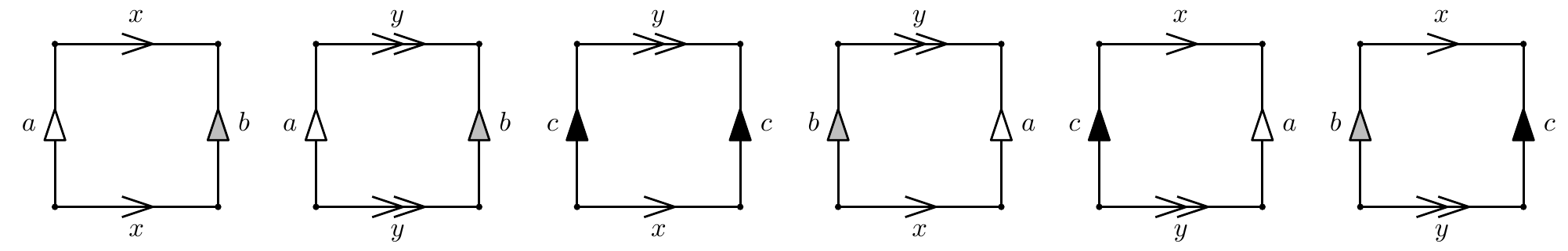}
  \end{center}
  \caption{The 6 squares defining the complex $\bX$}
  \label{fig-squares}
\end{figure}

The universal cover $\tH_X$ of $H_X$ is the 4-regular
infinite tree $F_4$. Its edges inherit the orientations from their
images in $H_X$: each vertex of $\tH_X$ has two incoming and
two outgoing arcs. Analogously, the universal cover $\tV_X$
of $V_X$ is the 6-regular infinite tree $F_6$ where each vertex has
three incoming and three outgoing arcs. Let $\tv_1$ be any vertex of
$\tH_X$. Then the principal filter of $\tv_1$ is the
infinite binary tree $T_2$ rooted at $\tv_1$: all its vertices except
$\tv_1$ have one incoming and two outgoing arcs, while $\tv_1$ has two
outgoing arcs and no incoming arc. Analogously, the principal filter
of any vertex $\tv_2$ in the ordered set $\tV_X$ is the
infinite ternary tree $T_3$ rooted at $\tv_2$.

Let $\tbX$ be the universal cover of $\bX$ and let $\varphi:
\tbX\rightarrow \bX$ be a covering map. Let $\tX$
denote the support of $\tbX$. Since $\bX$ is a CSC, by
\cite[Theorem 3.8]{Wi_csc}, $\tX$ is the Cartesian product
$F_4\times F_6$ of the trees $F_4$ and $F_6$.  The edges of
$\tbX$ are colored and oriented as their images in $\bX$,
and are also classified as horizontal or vertical edges. The squares
of $\tbX$ are oriented as their images in ${\bf X}$, thus
two opposite edges of the same square of $\tbX$ have the
same direction. This implies that all classes of parallel edges of
$\tbX$ are oriented in the same direction.  Denote this
orientation of the edges of $\tbX$ by $\tildo$. The
1-skeleton $\tX^{(1)}$ of $\tX$ together with
$\tildo$ is a directed median graph. Let
$\tv=(\tv_1,\tv_2)$ be any vertex of
$\tX$, where $\tv_1$ and $\tv_2$ are the
coordinates of $\tv$ in the trees $F_4$ and $F_6$. Then the
principal filter ${\mathcal
  F}_{\tildo}(\tv,\tX^{(1)})$ of $\tv$ is the
Cartesian product of the principal filters of $\tv_1$ in $F_4$
and of $\tv_2$ in $F_6$, i.e., is isomorphic to $T_2\times T_3$.

By Lemma \ref{directed-median}, the orientation of the edges of ${\mathcal
  F}_{\tildo}(\tv,\tX^{(1)})$ corresponds to the
canonical basepoint orientation of ${\mathcal
  F}_{\tildo}(\tv,\tX^{(1)})$ with $\tv$ as the
basepoint. Moreover, by Proposition \ref{regular-finite-npc},
${\mathcal F}_{\tildo}(\tv,\tX^{(1)})$ is the domain
of a regular event structure with one isomorphism type of principal
filters. We summarize this in the following result:

\begin{lemma} For any vertex $\tv$ of $\tX$,  ${\mathcal F}_{\tildo}(\tv,\tX^{(1)})$ is the domain of a regular event structure with one isomorphism class of futures.
\end{lemma}

\subsection{Aperiodicity of $\tbX$}

We recall here the main properties of $\tbX$ established in \cite[Section 5]{Wi_csc}. Let $\tv=(\tv_1,\tv_2)$
be an arbitrary vertex of $\tbX$, where $\tv_1$ and $\tv_2$ are defined as before. From the definition of the covering map,
the loop of $\bX$ colored $y$ gives rise to a bi-infinite horizontal path $P_y$ of $\tbX^{(1)}$ passing via $\tv$ and whose all edges
are colored $y$ and are directed from left to right. Analogously, there exists a  bi-infinite vertical  path $P_c$ of $\tbX^{(1)}$
passing via $\tv$ and whose all edges are colored $c$ and are directed from bottom to top.

The projection of $P_y$ on the horizontal factor $F_4$ is a
bi-infinite path $P^h$ of $F_4$ passing via $\tv_1$. Analogously, the
projection of $P_c$ on the vertical factor $F_6$ is a bi-infinite path
$P^v$ of $F_6$ passing via $\tv_2$. Consequently, the convex hull
$\conv(P_y \cup P_c)$ of $P_y\cup P_c$ in the graph
$\tbX^{(1)}$ is isomorphic to the Cartesian product of
$P^h\times P^v$ of the paths $P^h$ and $P^v$. Therefore the subcomplex
of $\tbX$ spanned by $\conv(P_y\cup P_c)$ is a directed plane
$\Pi_{yc}$ tiled into squares (recall that each square is of one of 6
types and its sides are colored by the letters $a,b,c,x,y$), see
Figure~\ref{fig-grid}.  Wise showed that the plane
$\Pi_{yc}$ is not tiled periodically by the preimages of the squares
of $\bX$.


\begin{theorem}[\!\!{\cite[Theorem 5.3]{Wi_csc}}] \label{anti-torus} The
  plane $\Pi_{yc}$ tiled into squares is not doubly periodic.
\end{theorem}

\begin{figure}
  \begin{center}
    \includegraphics[page=11,scale=0.7]{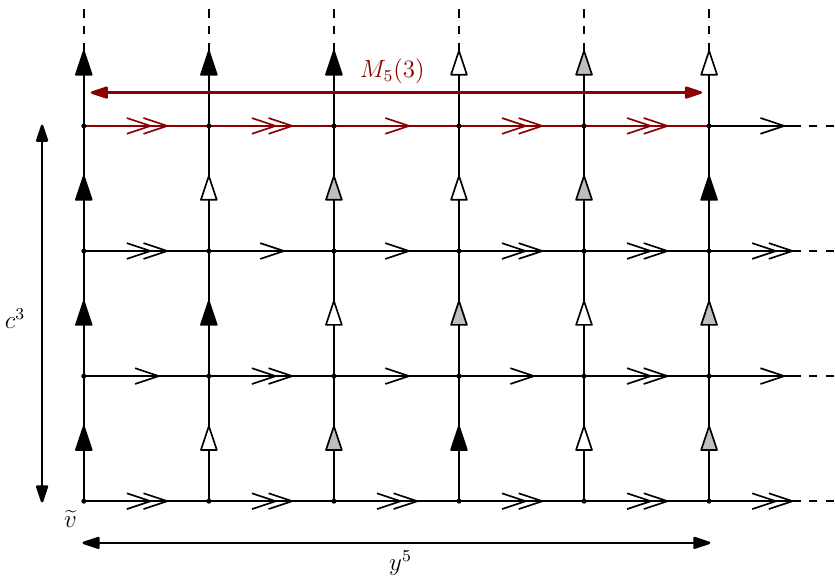}
  \end{center}
  \caption{Part of the plane $\Pi_{yc}^{++}$ appearing in
    $\tbX$}
  \label{fig-grid}
\end{figure}

In our counterexample we will use the following result of
\cite{Wi_csc} that was used to show that the plane $\Pi_{yc}$ is not
tiled periodically by the preimages of the squares of $\bX$.  Denote
by $P^+_y$ the (directed) subpath of $P_y$ having $\tv$ as a
source (this is a one-infinite horizontal path). Analogously, let
$P^+_c$ be the (directed) subpath of $P_c$ having $\tv$ as a
source. The convex hull of $P^+_y\cup P^+_c$ is a quarter of the plane
$\Pi_{yc}$, which we  denote by $\Pi^{++}_{yc}$. Any shortest path
in $\tbX^{(1)}$ from $\tv$ to a vertex $\tu\in
\Pi^{++}_{yc}$ can be viewed as a word in the alphabet $A=\{
a,b,c,x,y\}$. For an integer $n\ge 0$, denote by $y^n$ the horizontal
subpath of $P^+_y$ beginning at $\tv$ and having length
$n$. Analogously, for an integer $m\ge 0$, denote by $c^m$ the
vertical subpath of $P^+_c$ beginning at $\tv$ and having length
$m$.  Let $M_n(m)$ denote the horizontal path of $\Pi^{++}_{yc}$ of
length $n$ beginning at the endpoint of the vertical path
$c^m$. $M_n(m)$ determines a word which is the label of the side
opposite to $y^n$ in the rectangle which is the convex hull of $y^n$
and $c^m$ (see Figure~\ref{fig-grid}). Let $M_n(m)$ also denote this
corresponding word.

\begin{proposition}[\!\!{\cite[Proposition 5.9]{Wi_csc}}] \label{period-doubling} For each $n$,
the words $\{ M_n(m): 0\le m\le 2^n-1\}$ are all distinct, and thus, every positive word in $x$ and $y$
of length $n$ is $M_n(m)$ for some $m$.
\end{proposition}

This proposition is called in \cite{Wi_csc} ``period doubling''.
It immediately establishes Theorem \ref{anti-torus} because
it shows that the period of the
infinite vertical strip of $\Pi^{++}_{yc}$ of width $n$ and
bounded on the left by the path $P^+_c$ has period $2^n$.
Alternatively, every positive word in $x$ and $y$ appears in $\Pi^{++}_{yc}$,
and thus $\Pi_{yc}$ cannot be periodic.

\subsection{The  square complex $W$ and its universal cover $\tW$} \label{WiseW}

Let $\beta \bX$ denote the first barycentric subdivision of $\bX$:
each square $C$ of $\bX$ is subdivided into four squares
$C_1,C_2,C_3,C_4$ by adding a middle vertex to each edge of $C$ and
connecting it to the center of $C$ by an edge. This way each edge $e$
of $C$ is subdivided into two edges $e_1,e_2$, which inherit the
orientation and the color of $e$. The four edges connecting the middle
vertices of the edges of $C$ to the center of $C$ are oriented from
left to right and from bottom to top (see the middle figure of Figure
\ref{fig-barycentric-subdivision}).  Denote the resulting orientation
by $o'$.  This way, $(\beta \bX,o')$ is a directed and colored square
complex. Again, denote by $\beta X$ the support of $\beta \bX$. The
universal cover $\tbetaX$ of $\beta X$ is the Cartesian
product $\beta F_4\times \beta F_6$ of the trees $\beta F_4$ and
$\beta F_6$, where $\beta F_4$ is the first barycentric subdivision of
$F_4$ and $\beta F_6$ is the first barycentric subdivision of
$F_6$. Additionally, $(\tbetaX,\tildo')$ is a
directed CAT(0) square complex. We assign a \emph{type} to each vertex
of $\tbetabX$: the preimage of the unique vertex of $\bX$
is of type $0$ and is called a $0$-{\it
  vertex}, the preimages of the middles of  edges of $\bX$ are of type
$1$ and are called $1$-{\it vertices}, and the preimages of centers of
squares of $\bX$ are of type $2$ and are called $2$-{\it vertices}.

\begin{figure}
  \begin{center}
    \includegraphics[page=2,scale=0.75]{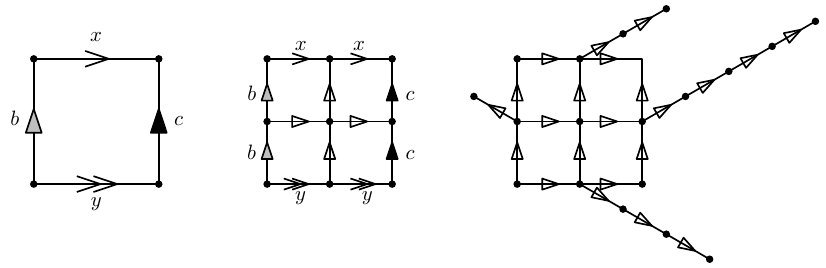}
  \end{center}
  \caption{A square of $\bX$ and the corresponding subcomplexes in
    $(\beta \bX,o')$ and  $(W,o^*)$}
  \label{fig-barycentric-subdivision}
\end{figure}

To encode the colors of the edges of $\bX$, we  introduce our
central object, the square complex $W$ (whose edges are no longer
colored). Let $A=\{ a,b,c,x,y\}$ and let $r: A\rightarrow \{
1,2,3,4,5\}$ be a bijective map. The complex $W$ is obtained from
$\beta X$ by adding to each 1-vertex $z$ of $\beta X$ a path $R_z$ of
length $r(\alpha)$ if $z$ is the middle of an edge colored $\alpha\in
A$ in $\bX$. The path $R_z$ has one end at $z$ (called the {\it root}
of $R_z$) and $z$ is the unique common vertex of $R_z$ and
$\beta X$ (we  call such added paths $R_z$ {\it tips}).

The square complex $W$ has 27 vertices: the unique vertex of $\bX$,
the 6 vertices which are the barycenters of the original squares, 5
vertices which are the barycenters of the original edges of $\bX$, and
15 vertices which are new vertices lying on tips. The complex $W$ has
$49$ edges: 10 corresponding to the 5 original edges that have been
subdivided, 24 connecting the barycenters of the original squares to
the barycenters of the original edges and 15 forming the tips. The
complex $W$ has 24 squares: 4 for each original square.

Denote by $o^*$ the orientation of the edges of $W$ defined as
follows: the edges of $\beta X$ are oriented as in $(\beta X,o')$ and
the edges of tips are oriented away from their roots (see the
rightmost figure of Figure \ref{fig-barycentric-subdivision} for the
encoding of the last square of Figure \ref{fig-squares}).  As a
result, we obtain a finite directed NPC square
complex $(W,o^*)$.

Consider the universal cover $\tW$ of $W$. It can be viewed
as the complex $\tbetaX$ with a path of length $r(\alpha)$
added to each 1-vertex which encodes an edge of $\tbX$ of
color $\alpha\in A$. We say that the vertices of $\tW$ lying
only on tips are of type $3$ and they are called $3$-\emph{vertices}.
Let $\tildo^*$ denote the orientation of the edges of
$\tW$ induced by the orientation $o^*$ of $W$. Then
$(\tW,\tildo^*)$ is a directed CAT(0) square complex.
Since $W$ is finite, by Proposition~\ref{regular-finite-npc}, the
directed median graph $(\tW^{(1)},\tildo^*)$ has a finite
number of isomorphisms types of principal filters ${\mathcal
  F}_{\tildo^*}(\tz,\tW^{(1)})$.

Let $\tv$ be any 0-vertex of $\tW$. Denote by
$\tW_{\tv}$ the principal filter ${\mathcal
  F}_{\tildo^*}(\tv,\tW^{(1)})$ of
$\tv$ in $(\tW^{(1)},\prec_{\tildo^*})$. By
Proposition \ref{regular-finite-npc}, $\tW_{\tv}$
together with the partial order $\prec_{\tildo^*}$ is the domain of
a regular event structure, which we  call {\it Wise's event
  domain}. Since vertices of different types of $\tW$ are
incident to a different number of outgoing squares, any isomorphism
between two filters of
$(\tW_{\tv},\prec_{\tildo^*})$ preserves the types of
vertices. We summarize all this in the following:

\begin{proposition} \label{wise-regular}  $(\tW_{\tv},\prec_{\tildo^*})$
is the domain of a regular event structure. Any isomorphism between any two filters of
$(\tW_{\tv},\prec_{\tildo^*})$ preserves the types of vertices.
\end{proposition}

\subsection{$({\tW_{\tv}},\prec_{\tildo^*})$ does not have a regular nice labeling}

In this subsection we prove that the event structure associated with
Wise's regular event domain is a counterexample to Thiagarajan's
conjecture.

\begin{theorem} \label{no-regular-labeling} $({\tW_{\tv}},\prec_{\tildo^*})$
does not admit a regular nice labeling.
Consequently, Conjectures~\ref{conj-thi} and~\ref{conj-thi-2} are false.
\end{theorem}

\begin{proof}   Since  ${\tW_{\tv}}$ is the  principal filter of a $0$-vertex $\tv$,  ${\tW_{\tv}}$ contains
all vertices of $\tbX$ located in
the quarter of plane  $\Pi^{++}_{yc}$ of $\tbX$, in particular it contains the vertices  of the
paths $P^+_c$ and $P^+_y$. Notice also that ${\tW_{\tv}}$
contains the barycenters and the tips corresponding to the edges  of $\Pi^{++}_{yc}$.

Suppose by way of contradiction that ${\tW_{\tv}}$ has a regular nice
labeling $\lambda$. Since ${\tW_{\tv}}$ has only a finite number of
isomorphism types of labeled filters, the vertical path $P^+_c$
contains two $0$-vertices, $\tz'$ and $\tz''$, which have isomorphic
labeled principal filters. Let $\tz'$ be the end of the vertical
subpath $c^k$ of $P^+_c$ and $\tz''$ be the end of the vertical
subpath $c^m$ of $P^+_c$, and suppose without loss of generality that
$k<m$. Let $n>0$ be a positive integer such that $m\le
2^n-1$. Consider the horizontal convex paths $M_n(k)$ and $M_n(m)$ of
$\Pi^{++}_{yc}$ of length $n$ beginning at the vertices $\tz'$ and
$\tz''$, respectively. For any $0\le i\le n,$ denote by $\tz_{k,i}$
the $i$th vertex of $M_n(k)$ (in particular,
$\tz_{k,0}=\tz'$). Analogously, denote by $\tz_{m,i}$ the $i$th vertex
of $M_n(m)$ (in particular, $\tz_{m,0}=\tz''$). In ${\tW_{\tv}}$, the
paths $M_n(k)$ and $M_n(m)$ give rise to two convex horizontal paths
$M^*_n(k)$ and $M^*_n(m)$ obtained from $M_n(k)$ and $M_n(m)$ by
subdividing their edges. Denote by $\tu_{k,i}$ the unique common
neighbor of $\tz_{k,i}$ and $\tz_{k,i+1}$, $0\le i<n$, in $M^*_n(k)$
(and in $\tW^{(1)}$). Analogously, denote by $\tu_{m,i}$ the unique
common neighbor of $\tz_{m,i}$ and $\tz_{m,i+1}$, $0\le i<n$ (see
Figure~\ref{fig-contre-ex}).  The paths $M^*_n(k)$ and $M^*_n(m)$
belong to the principal filters ${\mathcal
  F}_{\tildo^*}(\tz',\tW^{(1)})$ and ${\mathcal
  F}_{\tildo^*}(\tz'',\tW^{(1)})$, respectively.

 \begin{figure}
   \begin{center}
     \includegraphics[page=10, scale=0.65]{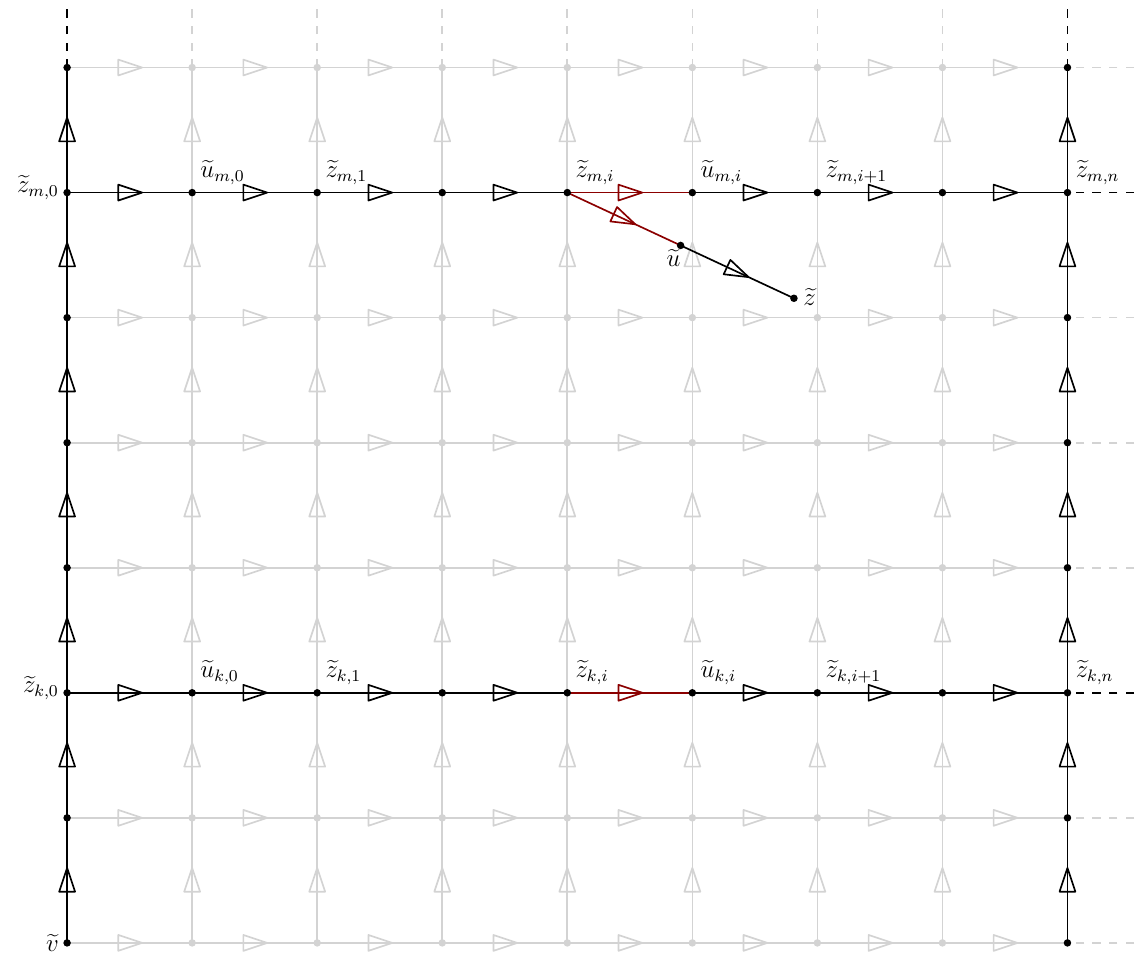}
   \end{center}
   \caption{To the proof of Theorem~\ref{no-regular-labeling}}
   \label{fig-contre-ex}
 \end{figure}

By  Proposition \ref{period-doubling}, the words $M_n(k)$ and $M_n(m)$ are different.   Let $f$ be an isomorphism between
the filters ${\mathcal F}_{\tildo^*}(\tz_{k,0},\tW^{(1)})$ and ${\mathcal F}_{\tildo^*}(\tz_{m,0},\tW^{(1)})$.
Since the words  $M_n(k)$ and $M_n(m)$  are different, from the choice of the lengths of tips  in the complexes $W$ and  $\tW$
it follows that $f$ cannot map  the path $M^*_n(k)$ to the path $M^*_n(m)$
by a vertical translation, i.e., there exists an index $0 \le j < n$
such that $f(\tz_{k,j+1})\ne \tz_{m,j+1}$; let $i$ be the smallest
such index.
Set $\tz:=f(\tz_{k,i+1})$ and $\tu:=f(\tu_{k,i})$. Since $f$
preserves the types of vertices, $\tz$ is a $0$-vertex and  $\tu$ is a 1-vertex. Since $f$ maps a convex path $M^*_n(k)$ to a
convex path, $\tu$ is the unique common neighbor of $\tz_{m,i}$ and $\tz$. Since each 1-vertex is the barycenter of a
unique edge of $\tbX$ and $\tz\ne \tz_{m,i+1}$, we deduce that
$\tu\ne \tu_{m,i}$. The edge $\tz_{k,i}\tu_{k,i}$ is directed from $\tz_{k,i}$ to $\tu_{k,i}$.
Analogously the edges  $\tz_{m,i}\tu_{m,i}$  and  $\tz_{m,i}\tu$  are directed from  $\tz_{m,i}$ to $\tu_{m,i}$
and  $\tu$, respectively.  Since $\tz_{k,i}\tu_{k,i}$  and  $\tz_{m,i}\tu_{m,i}$  are parallel edges, they
define the same event and therefore $\lambda(\tz_{k,i}\tu_{k,i})=\lambda(\tz_{m,i}\tu_{m,i})$. On the other hand,
since  $f$ maps the edge  $\tz_{k,i}\tu_{k,i}$ to the edge
$\tz_{m,i}\tu$ and since the map
$f$ preserves the labels,  we have  $\lambda(\tz_{k,i}\tu_{k,i})=\lambda(\tz_{m,i}\tu)$. As a result, $\tz_{m,i}$ has two outgoing edges,
 $\tz_{m,i}\tu_{m,i}$ and $\tz_{m,i}\tu$, having the same label, contrary to the assumption that $\lambda$
is a nice labeling. This contradiction shows that
$({\tW_{\tv}},\prec_{\tildo^*})$ does not admit a regular
nice labeling. By Proposition \ref{regular-finite-npc}, $({\tW_{\tv}},\prec_{\tildo^*})$ is the domain of a
regular event structure, establishing that Conjectures~\ref{conj-thi} and~\ref{conj-thi-2} are false.
This concludes the proof of the theorem.
\end{proof}

\subsection{$({\tW_{\tv}},\prec_{\tildo^*})$ has bounded $\natural$-cliques} \label{recognizable}

In this section, we show that our counterexample to Thiagarajan's
conjecture also provides a counterexample to
Conjecture~\ref{conj-badara-2} (and thus to
Conjecture~\ref{conj-badara}) of Badouel et
al~\cite{BaDaRa}. In~\cite{BaDaRa}, the conjecture was stated for
conflict event domains that are more general than the domain of event
structures we consider in this paper. However, we show in the next
proposition that their conjecture does not hold even for the domains
of event structures.



\begin{proposition}\label{cex-badara}
  Wise's event domain
  $(\tW_{\tv},\prec_{\tildo^*})$ has
  bounded $\natural$-cliques. Consequently,
  $(\tW_{\tv},\prec_{\tildo^*})$ is a
  counterexample to Conjectures~\ref{conj-badara}
  and~\ref{conj-badara-2}.
\end{proposition}

\begin{proof}
  By Proposition \ref{regular-finite-npc},
  $({\tW_{\tv}},\prec_{\tildo^*})$ is the domain of
  a regular event structure.  Recall that each event corresponds to a
  class of parallel edges of $\tW_{\tv}^{(1)}$.
  We refer to the events of
  $({\tW_{\tv}},\prec_{\tildo^*}^*)$ as {\it
    vertical, horizontal}, and {\it tip-events} depending of the type
  of edges from their parallelism class.


\begin{claimet}\label{claim-natural}
 If $e_1\natural e_2$ and  $e_1$ and $e_2$ are either both vertical or
 both horizontal, then $e_1\#_{\mu} e_2$.
\end{claimet}

\begin{proof}
  Without loss of generality, assume that both events $e_1$ and $e_2$
  are vertical, and note that $e_1$ and $e_2$ cannot be concurrent.
  Suppose by way of contradiction that $e_1\natural_{(3)} e_2$. Then
  there exists an event $e_3$ such that $e_1\| e_3$, $e_2\#_{\mu}e_3$
  and $e_3$ is co-initial with $e_1$ and $e_2$ at two different
  configurations. Since $e_1\| e_3$ and $e_1$ is vertical, the event
  $e_3$ cannot be vertical or a tip-event. Hence $e_3$ is
  horizontal. From the definition of ${\tW_{\tv}}$ it follows that the
  horizontal and vertical edges come from the Cartesian product of two
  trees. Therefore any pair of horizontal and vertical events defines
  a square of ${\tW_{\tv}}$, thus they are concurrent. This
  contradicts the fact that $e_3\#_{\mu} e_2$ and establishes the
  claim.
\end{proof}

Let $Q$ be a $\natural$-clique of ${\tW_{\tv}}$. We asserts that the
size of $Q$ is at most 11. Suppose that $|Q|\ge 12$.  From the
definition of $(\tW,\tildo^*)$ it follows that
$({\tW_{\tv}},\prec_{\tildo^*})$ has degree 5: the out-degree of any
$0$-vertex is 5, the out-degree of any 1-vertex is either 4 or 5, the
out-degree of any 2-vertex is 2, and the out-degree of any 3-vertex is
either 0 or 1. This implies that the maximum number of events of $Q$
that are pairwise concurrent or in minimal conflict is 5. From the
definition of $({\tW_{\tv}},\prec_{\tildo^*})$ it also follows that
two tip-events cannot be concurrent or in minimal conflict. Also from
condition (3) in the definition of $\natural$ it immediately follows
that $Q$ cannot contain two tip-events $e_1$ and $e_2$ such that
$e_1\natural_{(3)} e_2$. Indeed, if this happen, then there exists an
event $e_3$ such that $e_1\| e_3$, thus $e_1$ and $e_3$ cannot be
tip-events. Consequently, the $\natural$-clique $Q$ contains at most
one tip-event. Since $|Q|\ge 12$, $Q$ contains at least 6 vertical or
horizontal events, say $Q$ contains a subset $Q'$ of 6 vertical
events. Since all events of $Q'$ are vertical, they are not pairwise
concurrent. Since $Q'$ is a $\natural$-clique and at most 5 events of
$Q'$ can be pairwise in minimal conflict, this implies that $Q'$ must
contain two events $e_1,e_2$ such that $e_1\natural_{(3)} e_2$. But
this is impossible by the claim.  Therefore
$({\tW_{\tv}},\prec_{\tildo^*})$ is a regular conflict event domain
with bounded $\natural$-cliques and bounded degree. Since by Theorem
\ref{no-regular-labeling} $({\tW_{\tv}},\prec_{\tildo^*})$ does not
admit a regular nice labeling, this shows that
Conjecture~\ref{conj-badara-2} is false.
\end{proof}

\begin{remark}\label{rem-csc-badara}
  In the proof of Proposition~\ref{cex-badara}, we use the fact that
  any pair of horizontal and vertical events are concurrent. This
  property holds because $\bX$ is a CSC (complete square
  complex). Note that the fact that $\bX$ is a CSC is not an essential
  property of $\bX$ in the proof of Theorem~\ref{no-regular-labeling}.

  Consequently, if we want to adapt the proof of
  Theorem~\ref{no-regular-labeling} to other square complexes to find
  other counterexamples to Thiagarajan's Conjecture~\ref{conj-thi},
  it may be sufficient to consider $VH$-complexes (see
  Section~\ref{aperiodic}), but in order to use the arguments in the
  proof of Proposition~\ref{cex-badara} to find other counterexamples
  to Badouel et al.'s Conjecture~\ref{conj-badara}, we need to
  consider complete square complexes.
\end{remark}

\section{Aperiodic tilings and regular event structures} \label{aperiodic}


Our counterexample $(\tW_{\tv},\prec_{\tildo^*})$ of a regular
2-dimensional event domain without a regular labeling heavily uses
the fact that the universal cover $\tbX$ of Wise's complex
$\bX$ \cite{Wi_csc} contains a particular aperiodic tiled plane (that
is called \emph{antitorus} by Wise).  In this section, we show that
the relationship between the existence of aperiodic planes and
nonexistence of regular labelings is more general. Namely, we explain
how to obtain other counterexamples from 4-way deterministic aperiodic
tile sets.


\emph{Tiles} (or \emph{Wang-tiles}) are unit squares with colored
edges. The edges of a Wang tile are called \emph{top} (or
\emph{North}), \emph{right} (or \emph{East}), \emph{bottom} (or
\emph{South}) and \emph{left} (or \emph{West}) edges in a natural
way. A {\it tile set} $T$ is a finite collection of Wang-tiles, placed
with their edges horizontal and vertical.  A \emph{tiling} is a
mapping $f: \mathbb{Z}^2 \to T$ that assigns a tile to each integer
lattice point of the plane. A tiling $f$ is \emph{valid} if every two
adjacent tiles have the same color on their common edge. Note that a
tile may not be rotated or flipped, i.e., each tile has a bottom-top
and left-right orientation.  A tiling $f$ is \emph{periodic} with
period $(a,b) \in \mathbb{Z}^2\setminus\{(0,0)\}$ if for every $(x,y)
\in \mathbb{Z}^2$, $f(x,y) = f(x+a,y+b)$.  If there exists a valid
periodic tiling with tiles of $T$, then there exists a valid
\emph{doubly periodic} tiling with tiles of
$T$~\cite{Robinson_tiling}, i.e., a valid tiling $f$ and two integers
$a, b > 0$ such that $f(x,y) = f(x+a,y) = f(x,y+b)$ for every $(x,y)
\in \mathbb{Z}^2$.  A tile set $T$ is called \emph{aperiodic} if there
exists a valid tiling with tiles of $T$, and there does not exist any
periodic valid tiling with tiles of $T$.

Let $T=\{ t_1,\ldots, t_n\}$ be a tile set. We consider
each tile $t_i$ as a unit square whose edges are directed and colored.
Suppose that each square $t_i$ has two vertical and two horizontal
edges and suppose that the horizontal and the vertical edges of all
squares are colored differently, i.e., the set of colors can be
partitioned into horizontal colors and vertical colors. The horizontal
edges are directed from left to right and the vertical edges are
directed from bottom to top.

A Wang tile set is said to be \emph{NW-deterministic}~\cite{KaPa}, if
within the tile set there does not exist two different tiles that have
the same colors on their top and left edges. \emph{NE-deterministic},
\emph{SW-deterministic}, and \emph{SE-deterministic} tile sets are
defined analogously. A Wang tile set is \emph{4-way
  deterministic}~\cite{KaPa} if it is NW-, NE-, SW-, and
SE-deterministic. Kari and Papasoglu \cite{KaPa} presented a 4-way
deterministic aperiodic tile set $T_{KP}$.

Given a 4-way deterministic set of tiles $T$, let $\bX(T) =
(X(T),o,\nu)$ be the finite square complex obtained by identifying all
the vertices and gluing together the squares of $T$ along the sides
which have the same color respecting their orientation.  Then $\bX(T)$
is a $VH$-complex that has a unique vertex. Consequently, the
universal cover $\tbX(T)$ of $\bX(T)$ is a CAT(0) $VH$-complex.
Denote by $W(T)$ the finite directed NPC complex derived from
$\bX(T)$ in the same way as the complex $W$ was derived from Wise's
complex $\bX$ in Subsection \ref{WiseW} (taking the first barycentric
subdivision and adding tips of different lengths to encode the
different colors). Let $({\tW(T)}_{\tv},\prec_{\tildo^*})$ denote the
2-dimensional event domain derived from $\tbX(T)$ in the same way as
$({\tW_{\tv}},\prec_{\tildo^*})$ was derived from $\tbX$. Since
$({\tW(T)_{\tv}},\prec_{\tildo^*})$ comes from the universal cover of
the finite directed NPC complex $W(T)$,
$({\tW(T)}_{\tv},\prec_{\tildo^*})$ is a strongly regular event
structure.  The following lemma establishes a connection between the
existence of valid tilings for 4-way deterministic tile sets and the
existence of directed planes in the universal covers of the derived
$VH$-complexes.

\begin{lemma}\label{lem-Ttiles}
  For a 4-way deterministic tile set $T$, the following conditions are
  equivalent:
  \begin{enumerate}[(i)]
  \item there exists a valid tiling with the tiles of $T$;
  \item the universal cover $\tbX(T)$ of the square complex
    $\bX(T)$ contains directed planes;
  \item the strongly regular domain
    $({\tW(T)_{\tv}},\prec_{\tildo^*})$ is not hyperbolic.
  \end{enumerate}
\end{lemma}

\begin{proof}
  The implication $(i) \Rightarrow (ii)$ is trivial and the
  implication $(ii) \Rightarrow (iii)$ follows from
  Lemma~\ref{hyp_median}.  Suppose now that
  $({\tW(T)_{\tv}},\prec_{\tildo^*}))$ is not hyperbolic. Then by
  Lemma~\ref{hyp_median}, for any integer $k$, the $VH$-complex
  ${\tW(T)_{\tv}}$ contains a square grid of size $2k \times 2k$. The
  following claim implies that in such a grid, we can find a $k \times
  k$ directed square grid in the directed $VH$-complex
  $({\tW(T)_{\tv}},\tildo^*))$.

  \begin{claimet}
    For any vertical (respectively, horizontal) edge $\te$ going from
    $\tu$ to $\tw$ and for any two squares $Q_1, Q_2$ in
    $({\tW(T)_{\tv}},\tildo^*))$ intersecting
    on $\te$, $\tu$ cannot be the sink of both horizontal (respectively,
    vertical) edges of $Q_1$ and $Q_2$ incident to $\tu$.
  \end{claimet}

  \begin{proof}
    By way of contradiction, assume that $u$ is the sink of the
    horizontal edges $\te_1 =\tu_1\tu$ of $Q_1$ and $\te_2 = \tu_2\tu$
    of $Q_2$. By Lemma~\ref{directed-median}, $\tu_1, \tu_2 \in
    I(\tv,\tu)$ and the median $\tm$ of $\tu_1$, $\tu_2$, and $\tv$ is
    adjacent to $\tu_1, \tu_2$ and at distance $2$ from
    $\tu$. Consequently, $\tu\tu_1\tm\tu_2$ is a square of
    ${\tW(T)_{\tv}}$ and thus of $\tW(T)$ but since $\tu\tu_1$ and
    $\tu\tu_2$ are horizontal edges, this contradicts the fact that
    $\tbX(T)$ is a $VH$-complex.
  \end{proof}

  Consequently, we can tile arbitrary large squares of the plane with
  the tiles of $T$. By a folklore compactness result from tiling
  theory, this implies that we can find a valid tiling of the plane
  with the tiles of $T$, concluding the proof of $(iii) \Rightarrow
  (i)$.
\end{proof}

Note that if $T$ is a 4-way deterministic aperiodic tile set, all the
directed planes of $\tbX(T)$ are tiled in an aperiodic way.  In the
case of the tile set of Wise~\cite{Wi_csc} from
Figure~\ref{fig-squares}, the CAT(0) square complex $\tbX$
contains aperiodic directed planes but it also contains some periodic
directed planes.

\begin{remark}\label{rem-tore-bazar}
  As explained in~\cite[Section~4]{KaPa}, the universal cover
  $\tbX(T)$ of the complex $\bX(T)$ derived from a tile set $T$ can
  contain periodic planes that are not directed. This may happen even
  if $T$ does not tile the plane or if $T$ is an aperiodic tile set.

  For these reasons, if $T$ does not tile the plane, the directed
  CAT(0) complexes $\tbX(T)$ and $\tW(T)$ are not necessarily
  hyperbolic, even if all principal filters
  $({\tW_{\tv}},\prec_{\tildo^*})$ are hyperbolic domains.
\end{remark}

We now explain how to derive a counterexample to Thiagarajan's
conjectures from any 4-way deterministic aperiodic tile set.

\begin{theorem}\label{th-cex-pavages}
  For any 4-way deterministic aperiodic tile set $T$, the NPC square
  complex $W(T)$ is not virtually special and the 2-dimensional event
  domain $({\tW(T)}_{\tv},\prec_{\tildo^*})$ does not admit a regular
  nice labeling.
\end{theorem}

\begin{proof}
  Consider a 4-way deterministic aperiodic tile set $T$ and the
  associated NPC square complexes $\bX(T)$ and $W(T)$. Since $T$ tiles
  the plane, every vertex $\tv \in \tbX(T)$ is contained in a directed
  colored plane $\bPi$ of $\tbX(T)$. Note that the support $\Pi$ of
  $\bPi$ is the product of a directed path containing only horizontal
  edges and of a directed path containing only vertical edges.
  Consequently, in the directed CAT(0) complex $(\tW(T),\tildo^*)$,
  every $0$-vertex $\tv$ is contained in a directed plane $\Pi^*$
  where $\Pi^*$ is the first barycentric subdivision of $\Pi$.
  Consequently, the directed CAT(0) complex
  $({\tW(T)}_{\tv},\tildo^*)$ contains a quarter of the directed plane
  $\Pi^*$ that we denote by $\Pi^{*++}$. Note that $\Pi^{*++}$ is the
  barycentric subdivision of a quarter of plane $\Pi^{++}$ of the
  directed plane $\Pi$.  Let $P_H^*$ be the horizontal path of
  $\Pi^{*++}$ containing $\tv$.

  Suppose that $({\tW(T)}_{\tv},\prec_{\tildo^*})$ admits a regular
  nice labeling $\lambda$.  This implies that there exist two
  $0$-vertices $\ty, \ty' \in V(P_H^*)$ that have isomorphic labeled
  principal filters.  Let $P^*_V$ and $P'^*_V$ be the vertical paths
  of $\Pi^{*++}$ containing respectively $\ty$ and $\ty'$. Let $P^*_V
  = (\ty = \ty_{0}, \tu_{0}, \ty_{1}, \tu_{1}, \ldots, \ty_{j},
  \tu_{j}, \ldots)$ and $P'^{*}_V = (\ty' = \ty_0', \tu_0', \ty_1',
  \tu_1', \ldots, \ty_j', \tu_j' \ldots)$. Note that for every $j$,
  $\ty_j$ and $\ty_j'$ are $0$-vertices while $\tu_j$ and $\tu_j'$ are
  $1$-vertices. Note that $P_V = (\ty = \ty_0, \ty_1, \ldots, \ty_j,
  \ldots)$ and $P_V' = (\ty' = \ty_0', \ty_1', \ldots, \ty_j',
  \ldots)$ are paths of $\Pi^{++}$.

  Note that for any $j$, $\ty_j\tu_j$ and $\ty_j'\tu_j'$ are parallel
  edges as well as $\tu_j\ty_{j+1}$ and $\tu_j'\ty_{j+1}'$.
  Consequently, $\lambda(\ty_j\tu_{i}) = \lambda(\ty_j'\tu_{i}')$ and
  $\lambda(\tu_j\ty_{j+1}) = \lambda(\tu_j'\ty_{j+1}')$. Since
  $\lambda$ is a nice labeling (and thus is deterministic), and since
  $\ty$ and $\ty'$ have isomorphic labeled principal filters, one can
  easily show by induction on $j$ that for any $j$, $\ty_j$ and
  $\ty_j'$ (respectively, $\tu_j$ and $\tu_j'$) have isomorphic
  labeled principal filters. Consequently, for any $j$, the tips
  attached to $\tu_j$ and $\tu_j'$ have the same length, i.e., the
  edges $\ty_j\ty_{j+1}$ and $\ty_j'\ty_{j+1}'$ have the same color
  $\nu(\ty_j\ty_{j+1}) = \nu(\ty_j'\ty_{j+1}')$ in $\tbX(T)$.

  Since $\lambda$ is a regular nice labeling of
  $({\tW(T)}_{\tv},\prec_{\tildo^*})$, there exists
  $0 \leq k < m $ such that $\ty_k$ and $\ty_m$ have isomorphic
  labeled principal filters. Let $P_k^*$ and $P_m^*$ be the
  horizontal paths of $\Pi^{*++}$ going respectively from $\ty_k$ to
  $\ty_k'$ and from $\ty_m$ to $\ty_m'$. Let $\ell$ be the distance
  from $\ty_k$ to $\ty_k'$ in $\tbX(T)$ and let $P_k^* = (\ty_k =
  \ty_{k,0}, \tu_{k,0}, \ty_{k,1}, \tu_{k,1}, \ldots, \tu_{k,\ell-1},
  \ty_{k,\ell} = \ty_k')$ and $P_m^* = (\ty_m = \ty_{m,0}, \tu_{m,0},
  \ty_{m,1}, \tu_{m,1}, \ldots, \tu_{m,\ell-1}, \ty_{m,\ell} =
  \ty_m')$. Note that $P_k = (\ty_k = \ty_{k,0}, \ty_{k,1}, \ldots,
  \ty_{k,\ell} = \ty_k')$ and $P_m = (\ty_m = \ty_{m,0}, \ty_{m,1},
  \ldots, \ty_{m,\ell} = \ty_m')$ are paths of the plane $\Pi$.  Using
  the same arguments as for $P_V^*$ and $P'^{*}_V$, one can show that
  for any $0\leq i \leq \ell-1$, the edges $\ty_{k,j}\ty_{k,j+1}$ and
  $\ty_{m,j}\ty_{m,j+1}$ have the same color
  $\nu(\ty_{k,j}\ty_{k,j+1}) = \nu(\ty_{k,j}'\ty_{k,j+1}')$ in
  $\tbX(T)$.

  Consider the rectangle $R$ of $\Pi$ with corners $\ty_k$, $\ty_m$,
  $\ty_m'$, and $\ty_k'$.  For any $k \leq j < m$,
  $\nu(\ty_j\ty_{j+1}) = \nu(\ty_j'\ty_{j+1}')$ in
  $\tbX(T)$, i.e., the same sequence of colors appears on both
  vertical sides of $R$. Similarly, the same sequence of colors
  appears on both horizontal sides of the rectangle $R$. Since we can
  tile the plane by using copies of $R$, it is possible to find a
  periodic tiling of the plane using tiles of $T$. But this is
  impossible, since $T$ is an aperiodic tile set.  Consequently, the
  2-dimensional event domain
  $({\tW(T)}_{\tv},\prec_{\tildo^*})$ does not
  admit a regular nice labeling, and by
  Theorem~\ref{virtuallyspecial}, $W(T)$ is not virtually special.
\end{proof}

Using the tile set $T_{KP}$ of \cite{KaPa}, Lukkarila \cite{Lu} proved that for 4-way deterministic
tile sets the tiling problem is undecidable. An immediate consequence
of this result and of Theorem~\ref{th-cex-pavages} is that there
exists an infinite number of counterexamples to
Conjecture~\ref{conj-thi}.

\begin{remark}
Note that the $VH$-complex $W(T)$ derived from a $4$-way deterministic
tile set $T$ is not necessarily a CSC complex. Consequently, we cannot
directly generalize the proof of Proposition~\ref{cex-badara} to show
that if $T$ is aperiodic, then
$({\tW(T)}_{\tv},\prec_{\tildo^*})$ is a
counterexample to Conjecture~\ref{conj-badara} (see
Remark~\ref{rem-csc-badara}).
\end{remark}

\section{Conclusions and open questions}

\subsection{Conclusions}
In this paper, we presented an example of a regular event domain
$({\tW_{\tv},\prec_{\tildo^*}})$ with bounded degree and bounded
$\natural$-cliques which does not admit a regular nice labeling,
providing a counterexample to Conjecture~\ref{conj-thi} of
Thiagarajan~\cite{Thi_regular,Thi_conjecture} and
Conjecture~\ref{conj-badara} of Badouel, Darondeau, and
Raoult~\cite{BaDaRa}.  Furthermore, we show that this counterexample
is not singular and that, in fact, there exists an infinite number of
counterexamples to Conjecture~\ref{conj-thi} arising from the 4-way
deterministic aperiodic tile sets constructed by
Kari-Papasoglu~\cite{KaPa} and Lukkarilla~\cite{Lu}.

The event domain $(\tW_{\tv},\prec_{\tildo^*})$ is a principal filter
of a directed 2-dimensional CAT(0) cube complex which is the universal
cover of a finite directed colored CSC.  At first, one can think that
after trees, such cube complexes are the next simplest event domains
on which Conjectures \ref{conj-thi} or \ref{conj-badara} must be
true. Moreover, it was shown in \cite{ChHa} that any 2-dimensional
CAT(0) cube complex of bounded degree admits a finite nice labeling. A
finite nice labeling of ${\tW_{\tv}}$ can be also directly derived
from the fact that $\tW$ is a product of two trees with attached tips
of various lengths at 1-vertices. However, it turned out that finding
a {\it regular} nice labeling is not always possible even in the case
of 2-dimensional event domains (even those arising from CSC).

On the positive side, we proved that Thiagarajan's conjecture is true
for cover-special event structures (recall that Nielsen and
Thiagarajan established this conjecture for conflict-free event
structures and Badouel et al. proved it for context-free event
structures). As a consequence of deep results from geometric group
theory by Agol and Haglund-Wise, we deduce that strongly hyperbolic
regular event structures are cover-special, showing that Thiagarajan's
conjecture holds for a large and natural class of event domains.

We conclude the paper with a list of open problems, and we hope that
some of these problems will be solved positively.

\subsection{Regular versus strongly regular event structures}
In view of Proposition \ref{regular-finite-npc},
any strongly regular event structure is regular. One can ask if the converse holds
(this was also mentioned by a referee of a preliminary version of this paper~\cite{CC-ICALP17}):

\begin{question} \label{universal-cover-regular} Is any regular event
structure strongly regular? 
\end{question}

A natural way to derive a finite directed NPC complex from the domain
$\mathcal D$ of a regular event structure $\mathcal E$ is to factorize
$\mathcal D$ over all equivalence classes of futures (i.e., to
identify in a single vertex all configurations having the same
principal filter up to isomorphism). Unfortunately, this construction
does not preserve the non-positive curvature of $\cD$. For example,
consider a domain $\cD$ as described on the left of
Figure~\ref{fig-quotient-marche-pas}. In the figure, only a part of
the domain is described: one has to imagine that the dashed arrows
lead to the remaining part of the domain with the assumption that two
nodes that have the same label have isomorphic principal filters. When
we factorize the domain $\cD$ over the equivalence classes of futures,
we obtain the square complex on the left of
Figure~\ref{fig-quotient-marche-pas}. Note that this square complex is
not an NPC square complex as it contains three squares that intersect
in a vertex and that pairwise intersect on edges and these three
squares do not belong to a $3$-cube.

\begin{figure}
  \includegraphics[page=13]{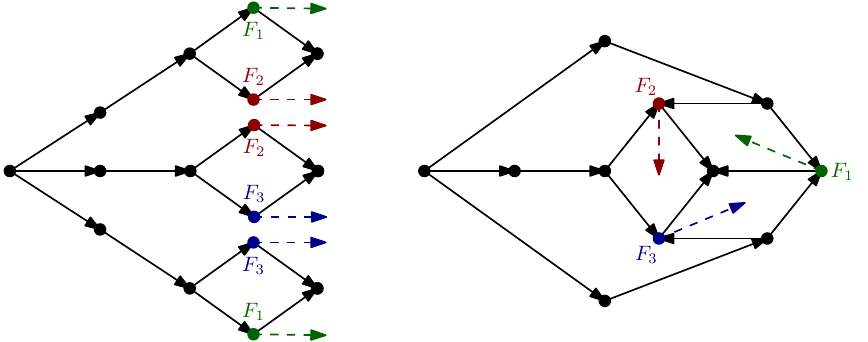}
  \caption{If we factorize the domain on the left over the equivalence
    classes of futures, we obtain the
    square complex on the right that is not an NPC square complex}
  \label{fig-quotient-marche-pas}
\end{figure}


This phenomenon does not arise if we consider $VH$-complexes and
isomorphisms that preserve vertical and horizontal edges. More
formally, the domain $\cD = \cD(\cE)$ of an event structure $\cE$ is a
\emph{$VH$-domain} if $\cD$ is a $VH$-complex. In this case,
$\cE$ is called a \emph{$VH$-event structure} and the events of $\cE$
are partitioned into vertical and horizontal events.  A $VH$-event
structure $\cE$ is \emph{$VH$-regular} if $\cE$ has finite degree and
has a finite number of principal filters up to isomorphism preserving
vertical and horizontal events. In this case, the domain $\cD(\cE)$ is
called a \emph{regular $VH$-domain}.

Even in this case, we do not know how to define formally a directed
NPC square complex according to the factorization mentioned above such
that the original domain is a principal filter of the universal cover
of this complex.

\begin{question} \label{universal-cover-regular-vh}
Does any regular $VH$-domain occur as a principal filter of the
universal cover of some finite directed $VH$-complex?
\end{question}

\subsection{Hyperbolic event domains} \label{ss:hyperbolic}
There are several natural reasons to investigate hyperbolic event
domains.  Similarly to CAT(0) and NPC spaces, Gromov hyperbolicity is
defined by a metric condition. However, similarly to the CAT(0)
property, the hyperbolicity of a CAT(0) cube complex can be expressed
in purely combinatorial way, by requiring that all isometric square
grids have bounded size.  Theorem \ref{hyperbolic} establishes that
Thiagarajan's conjecture is true for strongly hyperbolic regular event
structures. We conjecture that this result can be generalized in the
following way:

\begin{conjecture}\label{conj-hypst}
  Any strongly regular event structure with a hyperbolic domain admits
  a regular nice labeling.
\end{conjecture}

 Conjecture \ref{conj-badara} was positively solved by Badouel et
 al. \cite{BaDaRa} for context-free domains, which are particular
 hyperbolic domains:

\begin{lemma}
  Any context-free graph $G=(V,E)$ is hyperbolic.
\end{lemma}

\begin{proof}
  Let $G=(V,E)$ be a graph of uniformly bounded degree and $v$ be an
  arbitrary root (basepoint) of $G$. Let $S_i=\{ x\in V: d_G(v,x)=i\}$
  denote the sphere of radius $i$ centered at $v$. A connected
  component $F$ of the subgraph of $G$ induced by $V\setminus S_i(v)$
  is called an {\it end} of $G$.  The vertices of $F\cap S_{i+1}(v)$
  are called {\it frontier points} and this set is denoted by
  $\Delta(F)$ \cite{MuSch} and called a {\it cluster}.  Let
  ${\Phi}(G)$ denote the set of all ends of $G$, i.e., the connected
  subgraphs of $G(V\setminus S_i(v))$, when $i$ ranges over the
  natural numbers.  An {\it end-isomorphism} between two ends $F$ and
  $F'$ of $G$ is a mapping $f$ between $F$ and $F'$ such that $f$ is a
  graph isomorphism and $f$ maps $\Delta(F)$ to $\Delta(F')$.  Then
  $G$ is called a {\it context-free graph} \cite{MuSch} if ${\Phi}(G)$
  has only finitely many isomorphism classes under end-isomorphisms.
  Since $G$ has uniformly bounded degree, each cluster $\Delta(F)$ is
  finite. Moreover, from the definition of context-free graphs follows
  that a context-free graph $G$ has only finitely many isomorphism
  classes of clusters, thus there exists a constant $\delta<\infty$
  such that the diameter of any cluster of $G$ is bounded by
  $\delta$. By \cite[Proposition 12]{ChDrEsHaVa} any graph $G$ whose
  diameters of clusters is uniformly bounded by $\delta$ is
  $\delta$-hyperbolic (in fact, $G$ is quasi-isometric to a tree).
\end{proof}

The following conjecture generalizes Theorem~\ref{hyperbolic}, the
results of \cite{BaDaRa} in the case of event structures considered in
this paper, and Conjecture~\ref{conj-hypst}.

\begin{conjecture}\label{conj-hyperbolic}
  Conjectures \ref{conj-thi} and \ref{conj-badara} are true for
  hyperbolic event domains.
\end{conjecture}



By Lemma \ref{hyp_median}, the $1$-skeleton $X^{(1)}$ of a CAT(0) cube
complex is hyperbolic if and only if all isometrically embedded square
grids are uniformly bounded. In the language of event structures, an
isometrically embedded $n\times n$ grid $H$ corresponds to a
conflict-free event structure defined by $2n$ distinct events
$e_1,\ldots,e_n, f_1,\ldots,f_n$ such that any two events $e_i,f_j$
are concurrent and any two events $e_i,e_j$ or $f_i,f_j$ are either
causally dependent or concurrent. The isometricity follows from the
fact that the events $e_1,\ldots,e_n, f_1,\ldots,f_n$ are pairwise
distinct. If this grid is embedded in a hypercube, then any two events
$e_i,e_j$ or $f_i,f_j$ are concurrent. On the other hand, if
$e_1\lessdot e_2\lessdot\cdots\lessdot e_n$ and $f_1\lessdot
f_2\lessdot\cdots\lessdot f_n$, then this grid is isometrically
embedded as a directed flat square grid. A {\it (directed) flat square grid} of
side $n$ (respectively, a {\it (directed) flat plane}) of a median
graph $G$ is a (directed) $n\times n$-grid $H$ (respectively,
${\mathbb Z}\times {\mathbb Z}$-grid) isometrically embedded in $G$
such that any two squares of $H$ sharing a common edge do not belong
to a common 3-cube of $G$.  Note that if $H$ is a flat square grid or
a flat plane of a median graph $G$, then $H$ is a locally-convex
subgraph of $G$, and by Lemma \ref{convex}, $H$ is a convex subgraph
of $G$. This shows that if $G$ contains a flat square grid of size
$n$, then the graph $\Gamma_{\|}$ of the concurrent relation $\|$
contains an induced complete bipartite subgraph $K_{n,n}$. In a median
graph not containing 3-cubes (i.e., 1-skeletons of 2-dimensional
CAT(0) cube complexes), each embedded grid or plane is flat.  We
continue with a stronger version of Conjecture \ref{conj-hyperbolic}.

\begin{conjecture}\label{conj-hyperbolic-bis}
Conjectures \ref{conj-thi} and \ref{conj-badara} are true for event
domains with uniformly bounded sizes of directed flat square grids.
\end{conjecture}

A first step to solve this question could be to consider event
structures such that the graph $\Gamma_{\|}$ does not admit induced
complete bipartite subgraphs $K_{n,n}$ with arbitrarily large $n$.

\subsection{Confusion-free domains}
As we noticed already, Conjecture \ref{conj-thi}  was positively solved by Nielsen
and Thiagarajan \cite{NiThi} for conflict-free event structures.  A
possible way to generalize this result is to consider confusion-free
domains.

Conflict-free event structures  can be viewed as the event structures for which the
minimal-conflict graph $\Gamma_{\#_\mu}$ is edgeless, i.e.,
each event of $\mathcal E$ is a connected component of $\Gamma_{\#_\mu}$. (Notice that
conflict-free domains are not hyperbolic because they may contain ${\mathbb Z}^n$ for any $n$.)
Therefore, one way to extend the result of \cite{NiThi} is to consider more complex minimal-conflict
graphs $\Gamma_{\#_\mu}$. One possible such extension is to consider the event structures whose
minimal-conflict graphs $\Gamma_{\#_\mu}$ are disjoint unions of cliques. Such event structures can
be viewed as an extension of confusion-free event structures. An event structure $\mathcal E$ is
{\it confusion-free} \cite{NiPlWi} if the reflexive closure of minimal conflict is transitive and
$e\#_{\mu}e'$ implies $\downarrow\! e\setminus \{ e\}=\downarrow\! e'\setminus \{ e'\}$ (we use
the definition from \cite[Proposition 2.4]{VaVoWi}). From the first condition it follows that
for a confusion-free event structure the graph $\Gamma_{\#_\mu}$ is a disjoint union of cliques.
Confusion-free event structures correspond to deterministic concrete data structures \cite{KaPl}
and to confusion-free occurrence nets \cite{NiPlWi}.

\begin{question}
 Do Conjectures \ref{conj-thi} and \ref{conj-badara} hold
  for confusion-free event structures? More generally, do they hold
  for event structures whose minimal-conflict graph $\Gamma_{\#_\mu}$
  is a disjoint unions of cliques?
\end{question}



\subsection{Undecidability questions}
We think that the relationship between the existence of aperiodic tile
sets and the nonexistence of regular nice labelings of the associated
event structures may help to prove some undecidability results. We
conjecture that one cannot decide if a regular event structure
satisfies Thiagarajan's conjecture:

\begin{conjecture}\label{conj-undecidability-reg}
  There does not exist an algorithm that, given a strongly regular
  event domain $\cD$, can determine whether or not $\cD$ admits a
  regular nice labeling.
\end{conjecture}

The intuition behind is that one can use Lukkarilla's
construction~\cite{Lu} to prove this conjecture.  As in the proof of
undecidability of the classical tiling
problem~\cite{Berger,Robinson_tiling}, the undecidability proof of
Lukkarila is based on a reduction from the {\it Turing machine halting
  problem}.  More precisely, for any Turing machine $\cM$, Lukkarila
constructs a 4-way deterministic tile set $T_{\cM}$ such that either
$T_{\cM}$ is an aperiodic tile set (this corresponds to the case when
the Turing machine $\cM$ does not halt), or $T_{\cM}$ does not tile
the plane (this corresponds to the case when the Turing machine $\cM$
halts). In the first case, by Theorem~\ref{th-cex-pavages}, the domain
$(\tW(T_{\cM})_{\tv},\prec_{\tildo^*})$ does not admit a regular nice
labeling. In the second case, by Lemma~\ref{lem-Ttiles},
$(\tW(T_{\cM})_{\tv},\prec_{\tildo^*})$ is a strongly regular domain that is
hyperbolic. Consequently, if Conjecture~\ref{conj-hypst} was true,
$(\tW(T_{\cM})_{\tv},\prec_{\tildo^*})$ would admit a
regular nice labeling. This would prove
Conjecture~\ref{conj-undecidability-reg}.

Another possible way to prove Conjecture~\ref{conj-undecidability-reg}
 would be to anwser the following question in a positive way and use
Theorem~\ref{virtuallyspecial}.

\begin{question}\label{q-notiling-special}
  Given a $4$-way deterministic tile set $T$ such that there is no
  valid tiling with the tiles of $T$, is it true that the $VH$-complex
  $W(T)$ is virtually special?
\end{question}

Note that if there was a positive answer to this question, this would
answer a question of Agol~\cite[Question 3]{Agol_ICM} and confirm the
following conjecture of Bridson and Wilton~\cite{BrWi}:

\begin{conjecture}[\!\!{\cite[Conjecture  1.2]{BrWi}}]
  There does not exist an algorithm that, given a finite NPC square
  complex $Y$, can determine whether or not $Y$ is virtually special.
\end{conjecture}

Indeed, in Lukkarila's construction, if the Turing machine $\cM$ does
not halt, then by Theorem~\ref{th-cex-pavages} $W(T_{\cM})$ is not
virtually special. On the other hand, if the Turing machine $\cM$
halts, then if the answer to Question~\ref{q-notiling-special} was
positive, $W(T_{\cM})$ would be virtually special.

\subsection*{Acknowledgements.}
 We are grateful to P.S. Thiagarajan for some email exchanges on
 Conjecture 1 and paper \cite{NiThi} and to our colleague R. Morin for
 several useful discussions.  The work on the results not presented in
 the preliminary version~\cite{CC-ICALP17} (Sections
 \ref{thiagu-special} and \ref{aperiodic}) was supported by ANR
 project DISTANCIA (ANR-17-CE40-0015).

\bibliographystyle{plainurl}
\bibliography{bib-thiagarajan}

\end{document}